\documentclass[twocolumn,10pt]{IEEEtran}

\usepackage{graphicx}
\usepackage{amsmath}
\usepackage{cite}
\usepackage{multirow}
\usepackage{subfigure}
\usepackage{stfloats} 
\usepackage{amssymb}

\usepackage{amsthm}
\theoremstyle{definition} \newtheorem{lemma}{Lemma}
\usepackage{color} 

\begin{document}

\title{High-Mobility Wideband Massive MIMO Communications: Doppler Compensation, Analysis and Scaling Law}

\author{Wei Guo, Weile Zhang, Pengcheng Mu, Feifei Gao, and Hai Lin
\thanks{\hspace{-0.0cm}
Part of this work has already been published in~\cite{Guo17Angle}.

W. Guo is with the State Key Lab of ISN, Xidian University, Xi'an, Shaanxi, 710071, China (e-mail: w.guo@xidian.edu.cn).

W. Zhang and P. Mu are with the School of Electronic and Information Engineering, Xi'an Jiaotong University, Xi'an, Shaanxi, 710049, China (e-mail: wlzhang@mail.xjtu.edu.cn, pcmu@mail.xjtu.edu.cn).

F. Gao is with the Tsinghua University, Beijing, 100084, China (e-mail: feifeigao@ieee.org).

H. Lin is with the Department of Electrical and Information Systems, Osaka Prefecture University, Osaka, Japan (email: lin@eis.osakafu-u.ac.jp).}
}
\maketitle

\vspace*{-0.2in}
\begin{abstract}
In this paper, we apply angle-domain Doppler compensation for high-mobility wideband massive multi-input multi-output (MIMO) uplink transmission. The time-varying multipath channel is considered between high-speed terminal and static base station (BS), where multiple Doppler frequency offsets (DFOs) are associated with distinct angle of departures (AoDs).
With the aid of the large-scale uniform linear array (ULA) at the transmitter, we design a beamforming network to generate multiple parallel beamforming branches, each transmitting signal pointing to one particular angle. Then, the transmitted signal in each branch will experience only one dominant DFO when passing over the time-varying channel, which can be easily compensated before transmission starts.
We theoretically analyze the Doppler spread of the equivalent uplink channel after angle-domain Doppler compensation, which takes into account both the mainlobe and sidelobes of the transmit beam in each branch. It is seen that the channel time-variation can be effectively suppressed if the number of transmit antennas is sufficiently large. Interestingly, the asymptotic scaling law of channel variation is obtained, which shows that the Doppler spread is proportional to the maximum DFO and decreases approximately as $1/\sqrt{M}$ ($M$ is the number of transmit antennas) when $M$ is sufficiently large. Numerical results are provided to corroborate the proposed scheme.
\end{abstract}


\section{Introduction}
High-mobility communications have gained much interest over the past couple of decades and have been considered as one of the key parts in high-speed railway (HSR)~\cite{Zhu13TDD,Wu16A,He16High,Wang16Channel}.
Compared with the conventional communications with low-mobility or static users, high-mobility communications face many new challenges, such as fast time-varying channels, frequent handovers and complex channel environments like viaducts, tunnels and mountain areas. The time variation of the fading channel is mainly caused by the Doppler frequency offsets (DFOs) arising from relative motion between the transceivers~\cite{Hlawatsch11Wireless}. The DFOs not only increase the complexity of channel estimation and equalization at the receiver but also introduce inter-carrier interference (ICI) when orthogonal frequency division multiplexing (OFDM) technique is applied in high-mobility wideband wireless communications~\cite{Hwang09OFDM,Wu16A}.

How to combat the DFOs has been widely studied in the literature. Since multiple DFOs have already been mixed together at the receiver, some works directly estimate the composite time-varying channel in either time or frequency domain~\cite{Giannakis98Basis,Wang18Channel,Berger10Application}. Considering that the multiple DFOs are closely associated with the angle of arrivals (AoAs) or angle of departures (AoDs) of multipaths, some pioneer works have been reported to mitigate the effect of the DFOs from spatial domain via small-scale antenna arrays~\cite{Zhang11Multiple,Guo13Multiple,Yang13Beamforming,Ng03Effect}.
In~\cite{Zhang11Multiple} and~\cite{Guo13Multiple}, the sparse channel is assumed, and the DFOs of the dominant multipaths are compensated with perfect knowledge of maximum DFO and angle information. Similarly, the DFOs of the line-of-sight (LoS) path from different base stations (BSs) are compensated via AoAs estimation in~\cite{Yang13Beamforming}. 
However, most of these works consider the sparse channel with few multipaths. When there are a large number of multipaths between transceivers, these methods may fail due to limited spatial resolution.

Recently, large-scale multi-input multi-output (MIMO), known as ``massive MIMO'', has been deemed as one of the key techniques in the fifth generation (5G) systems~\cite{Lu14An,Xie17Unified,Tan17Spectral,Zhao18Time,Lin17A}. In a typical massive MIMO system, a large-scale antenna array is configured at the transmitter or receiver to provide many additional benefits such as high-spatial resolution and high-spectral efficiency.
Some works have already demonstrated the potential of massive MIMO to deal with high-mobility related challenges \cite{Liu14On,Chen17Directivity,Fan15Doppler}. 
In~\cite{Liu14On}, the large-scale antenna arrays are employed at both the BS and the train to maximize the uplink capacity in HSR communications.
In~\cite{Chen17Directivity}, the massive MIMO beamforming schemes are presented with the location information for HSR transmissions. However, these works ignore the effect of multiple DFOs and are mainly applied for scenarios with dominant LoS, e.g., viaduct scenarios.
A joint estimation scheme for both AoA and DFO is proposed in~\cite{Fan15Doppler} via the large-scale receive antennas. Unfortunately, the estimator in~\cite{Fan15Doppler} works with single path single DFO and cannot be directly extended to multipath scenarios.

To solve multiple DFOs under high-mobility, Chizhik~\cite{Chizhik04Slowing} first pointed out that channel time-variation can be slowed down through beamforming provided by a large number of antennas. Based on this observation, some works have investigated the frequency synchronization problems from the angle-domain signal processing~\cite{Guo17High,Ge17High,You17BDMA,Zhang18Frequency}.
By exploiting the high-spatial resolution provided by a large-scale antenna array, the work in~\cite{Guo17High} proposes a systemic receiver design for high-mobility downlink transmission with DFOs estimation and compensation. The work in \cite{Ge17High} further studies the impact of array error on designing angle-domain DFO compensation. 
In~\cite{You17BDMA}, per-beam synchronization in time and frequency has been applied to the millimeter-wave (mmW)/Terahertz (THz) transmissions with a large number of antennas, where Doppler spread analysis is provided under the ideal assumption that the signals have been constrained within the narrow beams without any leakage.


In this paper, we consider the high-mobility uplink wideband transmission. After analyzing the main difference between downlink and uplink transmissions, we apply Doppler compensation in angle domain at the transmitter for high-mobility OFDM uplink.
Thanks to the high-spatial resolution provided by a large-scale uniform linear array (ULA), we generate multiple parallel branches through matched filter beamforming. Since the transmitted signal in each branch is affected by one dominant DFO, we can easily compensate the DFO before transmission starts. Based on this mechanism, we theoretically analyze the Doppler spread for the equivalent uplink channel when considering the effect from both mainlobe and sidelobes of the beam. The asymptotic scaling law of channel variation is provided, which shows that: 1) the Doppler spread is proportional to the maximum DFO; 2) the Doppler spread decreases approximately as $1/\sqrt{M}$ ($M$ is the number of transmit antennas) with increasing $M$ when $M$ is sufficiently large. Simulation results also show the effectiveness of the proposed scheme.

The rest of this paper is organized as follows. In Section II, the time-varying channel and signal models for high-mobility uplink transmission are introduced. In Section III, a new transmission scheme for DFO compensation is developed. The Doppler spread for the equivalent uplink channel is analyzed in Section IV. Simulation results are provided in Section V, and Section VI concludes this paper.

\emph{Notations}: $(\cdot)^\mathrm{H}$, $(\cdot)^\mathrm{T}$, $(\cdot)^\ast$, and $\mathrm{E}\{\cdot\}$ represent the Hermitian, transposition, conjugate and expectation operations, respectively. $\|\cdot\|$ denotes the Euclidean norm of a vector, whereas $|\cdot|$ denotes the absolute value of a scalar. $\lfloor\cdot\rfloor$ and $\lceil\cdot\rceil$ denote the integer floor and integer ceiling, respectively.

\section{System Model}\label{sec:ULSysMod}
Consider the scenarios of high-mobility uplink transmission, as shown in Fig.~\ref{fig:ULSysMod}, where a relay station (RS) is mounted on top of the train for decoding and forwarding the data between users and the BS~\cite{He16High}. The RS can avoid the high penetration loss caused by the train carriages and can perform the handovers in a group manner. Since the RS-to-BS link faces severe DFOs while the user-to-RS link is less affected by DFOs, we only consider the RS-to-BS link in this paper.
At the RS, a large-scale ULA is configured along the direction of the motion, i.e., X-axis in Fig.~\ref{fig:ULSysMod}.
\begin{figure}[t!]
\centering
\includegraphics[scale=0.35]{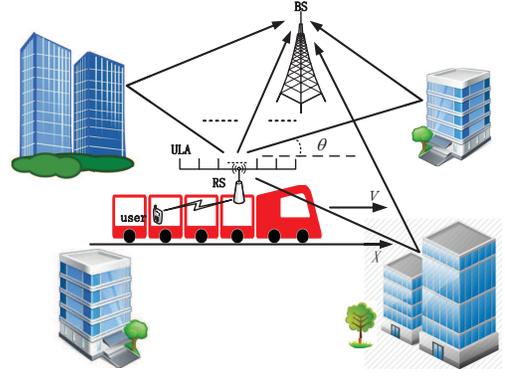}
\caption{High-mobility uplink transmission with rich reflectors.}
\label{fig:ULSysMod}
\vspace*{-0.0in}
\end{figure}

\vspace*{-0.0in}
\subsection{Time-Varying Multipath Channel Model}\label{ssec:ULChnMod}
We consider the transmission scenarios with rich reflectors, such as urban districts, mountain areas, and tunnels. Fig.~\ref{fig:ULSysMod} shows an example of the uplink transmission in urban district: A moving train travels through various buildings and the transmitted signal spreads from the RS to the BS through multipaths. 
Assume that the large-scale ULA has $M$ transmit antennas and the conventional small-scale antenna array with $N$ receive antennas is configured at the BS. Take the channel from the $n_t$th transmit antenna to the $n_r$th receive antenna as an example. The baseband time-varying multipath channel for the antenna pair $\{n_r,n_t\}$ can be modeled as
\begin{equation}\label{eq:h_chn}
  h_{n_r,n_t}\left(n,n'\right) = \sum\limits_{l=1}^{L}{g_{n_r,n_t}(l,n-d_l)\delta\left(n'-d_l\right)},
\end{equation}
where $L$ is the total number of channel taps with different delays, $d_l$ is the relative delay of the $l$th tap, and $g_{n_r,n_t}(l,n)$ is the corresponding complex amplitude for the antenna pair $\{n_r,n_t\}$.
To characterize the scenarios with rich reflectors, the Jakes' channel model has been widely used in the literature~\cite{Jakes94Microwave} and an established simulator has been proposed in~\cite{Zheng03Simulation}, where each tap is comprised of $N_p$ propagation paths. The equivalent model for $g_{n_r,n_t}(l,n)$ is given by
\begin{equation}\label{eq:h_tap}
  g_{n_r,n_t}(l,n) = \sum\limits_{q=1}^{N_p} {\alpha_{l,q} e^{j[2\pi{f_d}n{T_s}\cos{\theta_{l,q}}+\psi_{n_t}(\theta_{l,q})+\psi'_{n_r}(\vartheta_{l,q})]} },
\end{equation}
where $\alpha_{l,q}$ is the random complex path gain associated with the $q$th propagation path in the $l$th tap, and $T_s$ is the sampling interval. Here, $\psi_{n_t}(\theta_{l,q})$ and $\psi'_{n_r}(\vartheta_{l,q})$ represent the phase shifts induced at the $n_t$th transmit antenna and the $n_r$th receive antenna, respectively, which depend on the antenna structure, position, and the direction of the path. Moreover, $\theta_{l,q}$ and $\vartheta_{l,q}$ are the AoD and AoA, respectively, of the $q$th path in the $l$th tap relative to the moving direction, which are randomly distributed between 0 and $\pi$. In \eqref{eq:h_tap}, $f_d$ is the maximum DFO defined as $f_d=v/\lambda$, where $v$ is the speed of moving terminal, and $\lambda$ is the wavelength of carrier wave. Denote $f_{l,q}$ as the DFO for the $q$th path in the $l$th tap. It is seen that $f_{l,q}$ is determined by the AoD $\theta_{l,q}$ and the maximum DFO $f_d$, i.e., $f_{l,q}={f_d}\cos{\theta_{l,q}}$.

Take the first transmit antenna as reference. Then, the phase shift, $\psi_{n_t}(\theta_{l,q})$, for the ULA can be expressed as
\begin{equation}
  \psi_{n_t}(\theta_{l,q})=2\pi(n_t-1)d_t\cos(\theta_{l,q})/\lambda,
\end{equation}
where $d_t$ is the transmit antenna element spacing. The steering vector for the whole transmit antenna array at the AoD $\theta_{l,q}$ is then determined as
\begin{equation}
  \mathbf{a}_t(\theta_{l,q})=[e^{j\psi_1(\theta_{l,q})},\cdots,e^{j\psi_{M}(\theta_{l,q})}]^\mathrm{T}.
\end{equation}



In a high-mobility environment, $g_{n_r,n_t}(l,n)$ varies with time index $n$ due to the significant DFOs. Each path in~\eqref{eq:h_tap} is assumed to have independent attenuation, phase, AoD, AoA and also DFO~\cite{Zheng03Simulation}. During the transmitting period of one OFDM frame, there is little change in the position and speed of moving terminal. Therefore, we can assume that $\alpha_{l,q}$, $\theta_{l,q}$, $\vartheta_{l,q}$, and $f_d$ are constant over the observed data frame, and may vary among different frames. Note that when there are rich scatters around the moving terminal, $N_p$ tends to be very large and the channel model \eqref{eq:h_tap} coincides with the classical Jakes' channel model~\cite{Jakes94Microwave}.

\vspace*{-0.0in}
\subsection{Signal Model}\label{ssec:ULSigMod}
Suppose each frame consists of $N_b$ OFDM blocks, and denote $\mathbf{x}_m=\left[ x_{m,0},x_{m,1},\cdots ,x_{m,N_c-1} \right]^\mathrm{T}$ as the information symbols in the $m$th OFDM block, where $N_c$ is the number of subcarriers. After applying an $N_c$-point inverse discrete Fourier transform (IDFT) operator and adding the cyclic prefix (CP) of length $N_{cp}$, the resulting time-domain samples in the $m$th block can be expressed as
\begin{equation}\label{eq:s_mn}
  s_m\left(n\right) =\frac{1}{\sqrt{N_c}}\sum\limits_{k=0}^{N_c-1}{x_{m,k}e^{j\frac{2\pi kn}{N_c}}},-N_{cp}\leq n\leq N_{c}-1.
\end{equation}


Without loss of generality, we assume the total transmit power is normalized and is equally allocated to $M$ transmit antennas. Then, the signal on the $n_t$th transmit antenna can be expressed as $\tilde{s}_{m,n_t}(n)=s_m(n)/\sqrt{M}$. We further denote $\tilde{\mathbf{s}}_{m}(n)=[\tilde{s}_{m,1}(n),\cdots,\tilde{s}_{m,M}(n)]^\mathrm{T}$ as the transmitted signal vector for the whole antenna array at the RS. Assume perfect time synchronization at the receiver. From~\eqref{eq:h_chn} and~\eqref{eq:s_mn}, the $n$th time-domain sample in the $m$th OFDM block at the $n_r$th receive antenna can be expressed as
\begin{align}\label{eq:y_mn}
  {y}_{m,n_r}(n) =& \sum\limits_{n_t=1}^{M}\sum\limits_{l=1}^{L} g_{n_r,n_t}(l,mN_s+n-d_l)\tilde{s}_{m,n_t}(n-d_l) \nonumber\\ +& {z}_{m,n_r}(n),
\end{align}
where $N_s=N_c+N_{cp}$ is the length of an OFDM block, and $z_{m,n_r}(n)$ is the corresponding time-domain sample of the complex additive white Gaussian noise (AWGN) at the $n_r$th receive antenna. 


Denote $\mathbf{y}_{m,n_r} = [y_{m,n_r}(0),\cdots,y_{m,n_r}(N_c-1)]^\mathrm{T}$ and $\mathbf{z}_{m,n_r}=[z_{m,n_r}(0),\cdots,z_{m,n_r}(N_c-1)]^\mathrm{T}$. From~\eqref{eq:y_mn}, $\mathbf{y}_{m,n_r}$ can be expressed as
\begin{equation}\label{eq:y_m}
  \mathbf{y}_{m,n_r} = \sum\limits_{l=1}^{L} \sum\limits_{q=1}^{N_p} \rho_{l,q,n_r} \mathbf{a}_t^\mathrm{T}(\theta_{l,q}) \mathbf{S}_m(d_l) \mathbf{\Phi}_m(l,q) + \mathbf{z}_{m,n_r},
\end{equation}
where $\rho_{l,q,n_r}=\alpha_{l,q}e^{j\psi'_{n_r}(\vartheta_{l,q})}$ relates to the channel gain and the phase shift, and $\mathbf{S}_{m}(d_l)=[\tilde{\mathbf{s}}_{m}(0-d_l), \cdots, \tilde{\mathbf{s}}_{m}(N_c-d_l-1)] \in \mathbb{C}^{M\times N_c}$ corresponds to the transmitted signal matrix after delay of $d_l$. Moreover, $\mathbf{\Phi}_m(l,q)$ represents the following $N_c \times N_c$ diagonal phase rotation matrix introduced by DFO $f_{l,q}$:
\begin{equation}
  \mathbf{\Phi}_m(l,q) = \mathrm{diag}\{\beta_{m,0}(l,q),\cdots,\beta_{m,N_c-1}(l,q)\},
\end{equation}
where $\beta_{m,n}(l,q)=e^{j2\pi f_{l,q}(m{N_s}+n-d_l){T_s}}$.

When there is only one propagation path from the RS to the BS, i.e., $L=1$ and $N_p=1$, a single DFO appears in~\eqref{eq:y_m} and the conventional single frequency offset estimation and compensation techniques can be applied~\cite{Moose94A}. However, the multipath channel would make multiple DFOs mixed at the BS. Note that the observed multiple DFOs at BS are now related to AoDs at the transmitter side rather than AoAs at the receiver side. Hence, it is quite difficult for BS to perform DFOs estimation and compensation even with a large-scale antenna array~\cite{Guo17High,Ge17High}.

\section{Transmitter Design for High-Mobility OFDM Uplink}\label{sec:TxDesign}

\subsection{Motivation and Transmitter Design}\label{ssec:chn_est}
Let us first illustrate why it is beneficial to exploit a large-scale ULA at the transmitter to cope with multiple DFOs during the uplink transmission. The reason lies in three aspects:

\emph{1) Multi-antenna techniques can separate multiple DFOs.} As discussed previously, the received signal at a single antenna consists of multiple DFOs. Since the DFOs are associated with AoDs, it is difficult to separate them with one single antenna in either time or frequency domain. The multi-antenna techniques can provide spatial resolution and thus have the potential to separate multiple DFOs in spatial domain. 

\emph{2) Only the transmitter can distinguish multiple DFOs in the uplink transmission.} As introduced in \cite{Guo13Multiple}, in the uplink transmission, one AoA may be associated with multiple DFOs while the AoD has a one-to-one relationship with the DFO. Therefore, the receiver cannot distinguish multiple DFOs according to the AoAs while only the transmitter can exploit the AoDs to distinguish multiple DFOs.

\emph{3) Conventional small-scale antenna arrays cannot provide enough spatial resolution.} Some previous works have adopted small-scale antenna arrays in high-mobility systems. They are efficient for sparse channels with very few multipaths between transceivers~\cite{Zhang11Multiple,Guo13Multiple,Yang13Beamforming}. However, when there are a large number of multipaths with distinct DFOs, conventional small-scale antenna arrays become powerless due to limited spatial resolution. In comparison, the large-scale antenna array could support high-spatial resolution and provide the opportunity to deal with multiple DFOs in angle domain.

\begin{figure}[t!]
\centering
\includegraphics[scale=0.55]{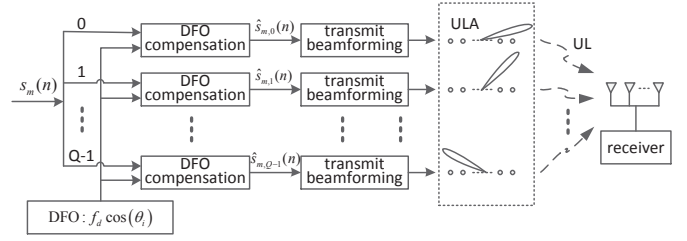}
\caption{Diagram of the transmitter design.}
\label{fig:ULTxDesign}
\vspace*{-0.0in}
\end{figure}

Based on the above observations, we propose to utilize a large-scale ULA at the transmitter to cope with the significant DFOs in high-speed transmission.
We generate multiple parallel beamforming branches simultaneously via the digital beamforming techniques. Here, each branch transmits signal towards the pre-determined direction through high-resolution beamforming. Since the transmitted signal in each branch is constrained within a narrow beam, it is mainly affected by one DFO when passing over the channel. Hence, it is possible to perform single DFO compensation at the transmitter for each branch.
Fig.~\ref{fig:ULTxDesign} shows the diagram of the transmitter design. The transmitted signal passes through $Q$ parallel branches simultaneously. In each branch, the DFO compensation and transmit beamforming are performed in sequence.


\vspace*{-0.0in}
\subsection{Beamforming Network and DFO Compensation}\label{ssec:bf}
As introduced in~\cite{Guo17High}, we place the large-scale ULA in the moving direction to make the distribution of the beam pattern in accordance with that of the DFOs. The antenna element spacing is selected as $d_t<\lambda/2$ to produce only one main beam in each transmit beamforming~\cite{Guo17High}. We further assume that $Q$ different values of $\theta$ are evenly selected between $0^\circ$ and $180^\circ$ to design the beamforming network, which are denoted by $\theta_i, i=0,1,\cdots,Q-1$. Here, the beamforming is performed in the whole angle range without estimating the AoDs of the multipaths. Hence, the beamforming network is designed in an offline manner to reduce the complexity. This is quite different from the existing competitors that depend on instantaneous angle information~\cite{Zhang11Multiple,Yang13Beamforming,Guo13Multiple}.

The goal of beamforming in each branch is to maintain the signal transmitted to the desired direction while suppressing the leakage to other directions. This can be easily implemented through the matched filter beamformer, whose weight vector for direction $\theta_i$ is determined by the steering vector
\begin{equation}\label{eq:w}
  \mathbf{w}(\theta_i) = \eta\mathbf{a}_t(\theta_i)e^{j\phi'(\theta_i)},
\end{equation}
where $\eta=\frac{1}{\left\|\sum_{i=0}^{Q-1}\mathbf{a}_t(\theta_i)e^{j\phi'(\theta_i)}\right\|}$ is the normalized parameter for restricting the transmitted signal power at multiple antennas, and $\phi'(\theta_i)$ is a random phase introduced at the direction $\theta_i$ to make the channel auto-correlation function independent of time and will be discussed in next section. 
Since the transmitted signal in each branch is affected by one dominant DFO after passing over the channel, we can easily compensate the DFO before transmission starts. Assume that the maximum DFO is known at the transmitter. For the $i$th branch, the main DFO induced by the relative motion is $f_d\cos{\theta_i}$ after transmit beamforming towards $\theta_i$. Then, the single DFO compensation can be easily performed as
\begin{equation}
  \hat{s}_{m,i}(n) = {s}_{m}(n) e^{-j2\pi f_{d}\cos\theta_i(m{N_s}+n){T_s}}.
\end{equation}

Regarding the beamforming towards $\theta_i$, the transmitted signal vector on the whole antenna array is given by
\begin{equation}
  \hat{\mathbf{s}}_{m,i}(n) = \mathbf{w}^*(\theta_i)\hat{s}_{m,i}(n).
\end{equation}

Similar to $\mathbf{S}_m(d_l)$ in~\eqref{eq:y_m}, define $\hat{\mathbf{S}}_{m,i}(d_l)=[\hat{\mathbf{s}}_{m,i}(0-d_l), \cdots, \hat{\mathbf{s}}_{m,i}(N_c-d_l-1)]$ as the transmitted signal matrix in the $m$th OFDM block of the $i$th branch after the delay of $d_l$. We further rewrite $\hat{\mathbf{S}}_{m,i}(d_l)$ in the following matrix form
\begin{equation}\label{eq:S_mi_2}
  \hat{\mathbf{S}}_{m,i}(d_l) = \mathbf{w}^*(\theta_i)\mathbf{s}_m(d_l)\mathbf{\Psi}_{m,i}(d_l),
\end{equation}
where $\mathbf{s}_{m}(d_l)=[s_{m}(0-d_l), \cdots, s_{m}(N_c-d_l-1)]$ is the original transmitted signal vector after the delay of $d_l$. Here, $\mathbf{\Psi}_{m,i}(d_l)$ stands for the DFO compensation matrix in the $i$th branch and can be expressed as
\begin{equation}
  \mathbf{\Psi}_{m,i}(d_l) = \mathrm{diag}\{\hat{\beta}_{m,0,i}(d_l),\cdots,\hat{\beta}_{m,N_c-1,i}(d_l)\},
\end{equation}
with $\hat{\beta}_{m,n,i}(d_l)=e^{-j2\pi f_d\cos\theta_i(m{N_s}+n-d_l){T_s}}$.

After replacing $\mathbf{S}_{m}(d_l)$ in~\eqref{eq:y_m} by $\hat{\mathbf{S}}_{m,i}(d_l)$, the received signal vector from the $i$th transmit branch can be written as
\begin{equation}\label{eq:r_mi}
  \mathbf{r}_{m,n_r,i} = \sum\limits_{l=1}^{L} \sum\limits_{q=1}^{N_p} \rho_{l,q,n_r} \mathbf{a}_t^\mathrm{T}(\theta_{l,q}) \hat{\mathbf{S}}_{m,i}(d_l) \mathbf{\Phi}_m(l,q).
\end{equation}

According to~\eqref{eq:S_mi_2}, we can further divide~\eqref{eq:r_mi} into
\begingroup\makeatletter\def\f@size{9.0}\check@mathfonts
\def\maketag@@@#1{\hbox{\m@th\normalsize\normalfont#1}}%
\begin{align}\label{eq:r_mi_1}
  & \mathbf{r}_{m,n_r,i} = \underbrace{\!\!\!\sum\limits_{l,q,\theta_{l,q}=\theta_i}\!\!\! \rho_{l,q,n_r}\eta{N_t}e^{-j\phi'(\theta_i)}\mathbf{s}_m(d_l)}_{desired\ signal} \nonumber\\ + & \underbrace{\sum\limits_{l',q',\theta_{l',q'}\neq\theta_i}\!\!\! \rho_{l',q',n_r}\mathbf{w}^\mathrm{H}(\theta_i) \mathbf{a}_t(\theta_{l',q'}) \mathbf{s}_m(d_{l'})\mathbf{\Psi}_{m,i}(d_{l'})\mathbf{\Phi}_m(l',q')}_{interference}.
\end{align}\endgroup

Here, the first term represents the desired signal from AoD $\theta_i$ which is equivalent to the standard received signal model through time-invariant frequency-selective channel. The second term is the interference from other directions and is affected by the residual DFOs. It can be observed that, with a sufficiently large $M$, $\mathbf{w}^\mathrm{H}(\theta_i)\mathbf{a}_t(\theta_{l',q'})\simeq 0$ holds for $\theta_i\neq\theta_{l',q'}$. Therefore, the second term has been greatly suppressed through the high-resolution transmit beamforming.

Finally, the received signal at the $n_r$th receive antenna is the sum of signals from $Q$ parallel branches, that is
\begin{equation}\label{eq:r_m_ant}
  \tilde{\mathbf{r}}_{m,n_r} = \sum\limits_{i=0}^{Q-1}\mathbf{r}_{m,n_r,i}+\mathbf{z}_{m,n_r}.
\end{equation}

Following the above discussions, the dominant DFO caused by high mobility has been compensated individually in each beamforming branch. Now the uplink channel can be considered as time-invariant approximately when the interference is mitigated to a tolerable magnitude. Especially, it is expected that the interference would turn to zero when $M$ approaches infinity. In this situation, the channel would become exactly time-invariant. Then, the conventional static channel estimation~\cite{Zhang11Multiple,Yang13Beamforming} and maximum ratio combing (MRC) detection can be carried out at BS.

\section{Doppler Spread Analysis}
%

As introduced above, with infinite transmit antennas, the time variation of channel can be completely eliminated. However, the number of transmit antennas is limited in practice, which means that the residual DFOs still appear at the receiver and would introduce a certain time variation.
In this section, we theoretically analyze the effect of residual DFOs in the proposed angle-domain DFO compensation scheme. We adopt the Doppler spread~\cite{Souden09Robust} to evaluate the time variation of the resultant uplink channel at BS.

\vspace*{-0.1in}
\subsection{PSD Analysis for the Equivalent Uplink Channel}\label{ssec:PSD}
We start with analyzing the power spectrum density (PSD) for the equivalent uplink channel which will show the contribution of each frequency component on channel variation~\cite{Jakes94Microwave,Souden09Robust}.
Since the channel taps are independent and have the same statistical property, we analyze the residual DFOs of only one tap for simplicity, i.e., $L=1$ and $d_1=0$. After omitting the receive antenna index $n_r$ and the noise term, the received signal in~\eqref{eq:r_m_ant} turns to
\begin{equation}\label{eq:r_m}
  \tilde{\mathbf{r}}_{m} =\!\! \sum\limits_{i=0}^{Q-1}\sum\limits_{q=1}^{N_p} \rho_{1,q} \mathbf{w}^\mathrm{H}(\theta_i) \mathbf{a}_t(\theta_{1,q}) \mathbf{s}_m(d_1)\mathbf{\Psi}_{m,i}(d_1) \mathbf{\Phi}_m(1,q).
\end{equation}

Assume that each tap has a large number of multipaths with AoDs between $0$ and $\pi$ and transmit beamforming is performed continuously between $0$ and $\pi$. Then, the summations in~\eqref{eq:r_m} can be replaced by integrations. After replacing $\theta_{1,q}$ and $\theta_{i}$ by $\tilde{\theta}$ and $\theta$, respectively, the equivalent uplink channel can be expressed in the continuous-time form as
\begin{equation}\label{eq:g_chn_bf}
\tilde{g}(t) \!=\! E_0 \int_0^{\pi} \int_0^{\pi}\!\! \alpha(\tilde{\theta})G(\theta,\tilde{\theta}) e^{j(\omega_d t y(\theta,\tilde{\theta})+\varphi(\theta,\tilde{\theta}) )} d\tilde{\theta} d\theta,
\end{equation}
where $t$ is the continuous-time index, $E_0$ is a scaling constant, $\alpha(\tilde{\theta})$ is the random gain for the path associated with AoD $\tilde{\theta}$, $\omega_d=2\pi f_d$, and $y(\theta,\tilde{\theta})=\cos{\tilde{\theta}}-\cos{\theta}$.
Here, $G(\theta,\tilde{\theta})$ stands for the antenna gain at direction $\tilde{\theta}$ when applying the matched filter beamformer towards direction $\theta$, and is determined as
\begin{equation}\label{eq:Gain}
  G(\theta,\tilde{\theta}) = \frac{1}{M}\mathbf{a}_t^\mathrm{H}(\theta)\mathbf{a}_t(\tilde{\theta})=\frac{\sin\left(\frac{\pi Md}{\lambda}y(\theta,\tilde{\theta})\right)} {M\sin\left(\frac{\pi d}{\lambda}y(\theta,\tilde{\theta})\right)}.
\end{equation}

In~\eqref{eq:g_chn_bf}, the random phase $\varphi(\theta,\tilde{\theta})$ is expressed as $\varphi(\theta,\tilde{\theta})=\phi(\tilde{\theta})+\phi'(\theta)$, where $\phi'(\theta)$ is the random phase artificially introduced for beamforming direction $\theta$ at the transmitter while $\phi(\tilde{\theta})$ is the random phase for the path associated with AoD $\tilde{\theta}$.
We introduce the random phase $\phi'(\theta)$ in the beamforming network such that the equivalent uplink channel satisfies stationary distribution. 

\begin{lemma}\label{Lemma:ChnCorr}
The autocorrelation for the equivalent fading channel in~\eqref{eq:g_chn_bf} is independent of the time index $t$ and can be expressed as a function of delay $\tau$, that is
\begin{equation}\label{eq:R}
R_{\tilde{g}\tilde{g}}(\tau) \!=\! E_0^2 \int_0^{\pi} \int_0^{\pi}\!\! \mathrm{E}\{\alpha^2(\tilde{\theta})\} |G(\theta,\tilde{\theta})|^2 e^{-j\omega_d \tau y(\theta,\tilde{\theta})} d\tilde{\theta} d\theta.
\end{equation}
\end{lemma}

\begin{proof}
See Appendix~\ref{Proof:ChnCorr}.
\end{proof}

Here, we assume $\int_0^{\pi}\mathrm{E}\{\alpha^2(\tilde{\theta})\}d\tilde{\theta}=1$ and $E_0=\sqrt{2}$ as in \cite{Zheng03Simulation}. Note that there exists a constant coefficient between~\eqref{eq:R} and the accurate channel autocorrelation which will not affect the subsequent analysis.
Since it is hard to get the close-form expression for the integrals in \eqref{eq:R}, we develop an effective approximation for \eqref{eq:R}. With a given beamforming angle $\theta$, the beam pattern of the ULA determined by \eqref{eq:Gain} is comprised of two parts: mainlobe and sidelobes, as shown in Fig.~\ref{fig:Beampattern}. Therefore, it is natural to divide the autocorrelation in \eqref{eq:R} into the following two parts according to the range of the mainlobe and sidelobes
\begin{equation}\label{eq:R_1}
R_{\tilde{g}\tilde{g}}(\tau) = R_{\tilde{g}\tilde{g}}^{(m)}(\tau)+\sum_{\substack{I_\mathrm{min}\leq i\leq I_\mathrm{max} \\ i\neq\{-1,0\}}}R_{\tilde{g}\tilde{g},i}^{(s)}(\tau),
\end{equation}
where
\begin{equation}\label{eq:R_m}
R_{\tilde{g}\tilde{g}}^{(m)}(\tau) \!=\! \frac{2}{\pi} \int_0^{\pi} \int_{\theta-\Delta_{l}}^{\theta+\Delta_{u}}  |G(\theta,\tilde{\theta})|^2 e^{-j\omega_d \tau y(\theta,\tilde{\theta})} d\tilde{\theta} d\theta
\end{equation}
and
\begin{equation}\label{eq:R_s}
R_{\tilde{g}\tilde{g},i}^{(s)}(\tau) \!=\! \frac{2}{\pi} \int_0^{\pi} \int_{\theta_i-\Delta_{l,i}}^{\theta_i+\Delta_{u,i}} \! |G(\theta,\tilde{\theta})|^2 e^{-j\omega_d \tau y(\theta,\tilde{\theta})} d\tilde{\theta} d\theta
\end{equation}
are related to the mainlobe and the $i$th sidelobe of the beam pattern, respectively. In \eqref{eq:R_m}, $\Delta_{l}$ and $\Delta_{u}$ are the angle offsets of the lower and upper null points, respectively, adjacent to the beamforming direction $\theta$. As shown in Fig.~\ref{fig:Beampattern}, the width of the mainlobe is jointly determined by $\Delta_{l}$ and $\Delta_{u}$. According to the antenna gain in \eqref{eq:Gain}, the null points of the mainlobe satisfy: $\frac{\pi Md}{\lambda}y(\theta,\tilde{\theta})=\frac{\pi Md}{\lambda}(\cos\tilde{\theta}-\cos\theta)=\pm\pi$. Thus, $\Delta_{l}$ and $\Delta_{u}$ can be expressed as
\begin{equation}
  \left\{ \begin{array}{l} \Delta_{l} = \theta-\arccos(\cos\theta+\frac{\lambda}{Md}),\\ \Delta_{u} = \arccos(\cos\theta-\frac{\lambda}{Md})-\theta. \end{array} \right.
\end{equation}

In \eqref{eq:R_s}, $\theta_i$ is the direction of the $i$th sidelobe when beamforming towards direction $\theta$, and $\Delta_{l,i}$ and $\Delta_{u,i}$ are the angle offsets of the lower and upper null points in the $i$th sidelobe, respectively, as illustrated in Fig.~\ref{fig:Beampattern}. From \eqref{eq:Gain}, we know $\theta_i$ satisfies
\begin{equation}
  \frac{\pi Md}{\lambda}(\cos\theta_i-\cos\theta)=\frac{(2i+1)\pi}{2},\ i\neq0,-1,
\end{equation}
that is,
\begin{equation}\label{eq:theta_i}
\cos\theta_i = \cos\theta+u_i,\ i\neq0,-1,
\end{equation}
with $u_i=\frac{(2i+1)\lambda}{2Md}$. Moreover, since both $\theta$ and $\theta_i$ lie between 0 and $\pi$, the range for the integer $i$ is determined as:  $\{i|i\in[I_\mathrm{min},I_\mathrm{max}],\ i\neq 0,-1\}$ with $I_\mathrm{max}=\lfloor\frac{2Md}{\lambda}-\frac{1}{2}\rfloor$ and $I_\mathrm{min}=\lceil-\frac{2Md}{\lambda}-\frac{1}{2}\rceil$.

Note that only the mainlobe is considered when calculating the DFOs in~\cite{You17BDMA}. In our work, we consider the impact from both mainlobe and sidelobes. In fact, the residual DFOs induced by the sidelobes cannot be neglected which will be verified in the following analysis.

\begin{figure}[t!]
\centering
\includegraphics[scale=0.6]{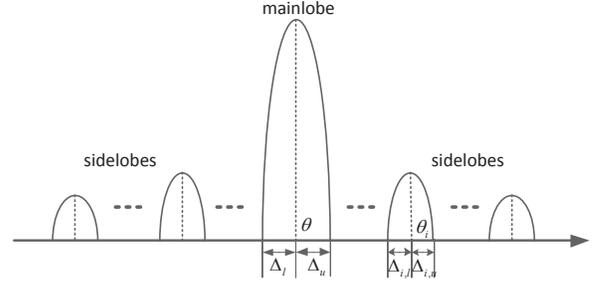}
\caption{An example of the beam pattern for the ULA.}
\label{fig:Beampattern}
\vspace*{-0.0in}
\end{figure}


\begin{lemma}\label{Lemma:R_m}
The autocorrelation of the fading channel related to the mainlobe in \eqref{eq:R_m} can be approximated as
\begin{align}\label{eq:R_m_3}
  R_{\tilde{g}\tilde{g}}^{(m)}(\tau) =& C_0 + 2C_1\frac{\sin(W_0\tau)}{\pi\tau} \nonumber\\ + & C_1\frac{\sin(W_0(\tau+t_0))}{\pi(\tau+t_0)} + C_1\frac{\sin(W_0(\tau-t_0))}{\pi(\tau-t_0)},
\end{align}
where $C_0=\frac{8\Delta_m\theta_t}{\pi}$, $C_1 = -\frac{2}{\omega_d}\ln\tan\frac{\theta_t}{2}$, $W_0=\frac{\omega_d\lambda}{Md}$, $t_0=\frac{\pi Md}{\omega_d\lambda}$, and $W_0t_0=\pi$. Here, $\Delta_m$ and $\theta_t$ are defined in \eqref{eq:Delta_m} and \eqref{eq:theta_t} in Appendix~\ref{Proof:R_m}, respectively.
\end{lemma}

\begin{proof}
See Appendix~\ref{Proof:R_m}.
\end{proof}

We see that the first term in \eqref{eq:R_m_3} is irrelevant to delay $\tau$ and is considered as the direct component. The second term is the sinc function of $\tau$. The third and fourth terms are the sinc functions after the delay of $-t_0$ and $t_0$, respectively.

\begin{lemma}\label{Lemma:R_s}
The autocorrelation of the fading channel related to the $i$th sidelobe in \eqref{eq:R_s} can be approximated as
\begin{align}\label{eq:R_s_0}
  R_{\tilde{g}\tilde{g},i}^{(s)}(\tau) =& D_i \bar{C}_i e^{-jW_i\tau}\left\lbrace\frac{2\sin(\frac{W_0}{2}\tau)}{\pi\tau} \right. \nonumber\\ +& \left.\frac{\sin(\frac{W_0}{2}(\tau+2t_0))}{\pi(\tau+2t_0)} + \frac{\sin(\frac{W_0}{2}(\tau-2t_0))}{\pi(\tau-2t_0)} \right\rbrace,
\end{align}
where $D_i=\left|\frac{1}{M\sin[\frac{(2i+1)\pi}{2M}]}\right|^2$, $W_i=u_i\omega_d$, and $\bar{C}_i$ is determined by the following functions
\begin{equation}\label{eq:Ci}
	\bar{C}_i = \left\{\! \begin{array}{ll} \frac{1}{\omega_d} \int_{\bar{\theta}_t}^{b_{i}}  \frac{1}{\sqrt{1-(\cos\theta_i-u_i)^2}}d\theta_i, & i>0,\\ \frac{1}{\omega_d} \int_{a_i}^{\pi-\bar{\theta}_t} \frac{1}{\sqrt{1-(\cos\theta_i-u_i)^2}}d\theta_i, & i<-1, \end{array} \right.
\end{equation}
with $\bar{\theta}_t$, $a_i$ and $b_i$ being defined in \eqref{eq:bartheta_t}, \eqref{eq:a_i} and \eqref{eq:b_i} in Appendix~\ref{Proof:R_s}, respectively.
\end{lemma}

\begin{proof}
See Appendix~\ref{Proof:R_s}.
\end{proof}

Finally, combining the discussions in both Lemma \ref{Lemma:R_m} and Lemma \ref{Lemma:R_s}, we obtain the close-form expression for the autocorrelation of the equivalent uplink channel in \eqref{eq:R_1}.

The PSD is defined as the Fourier transformation of channel autocorrelation~\cite{Jakes94Microwave,Souden09Robust}, that is
\begin{equation}\label{eq:PSD_def}
P(\omega) = \int_{-\infty}^{+\infty} R_{\tilde{g}\tilde{g}}(\tau)e^{-j\omega\tau} d\tau.
\end{equation}

\begin{lemma}\label{Lemma:PSD}
The PSD for the equivalent uplink channel can be written in the following form
\begin{align}\label{eq:PSD}
P(\omega) &= C_0\delta(\omega)+2C_1(1+\cos\omega t_0)X(j\omega) \nonumber\\
+& \!\!\! \sum_{\substack{I_\mathrm{min}\leq i\leq I_\mathrm{max} \\ i\neq\{-1,0\}}} \!\!\!\!\! 2{D_i}\bar{C}_i(1+\cos\left(2(\omega+W_i)t_0\right))\bar{X}(j(\omega+W_i)),
\end{align}
where $\delta(\omega)$ is the impulse response function, and $X(j\omega)$ and $\bar{X}(j\omega)$ are the rectangular window functions with the width of $W_0$ and $W_0/2$, respectively, that is
\begin{equation*}
X(j\omega) = \left\{\! \begin{array}{ll} 1, & |\omega|<W_0, \\ 0, & |\omega|>W_0, \end{array} \right.
\end{equation*}
\begin{equation*}
\bar{X}(j\omega) = \left\{\! \begin{array}{ll} 1, & |\omega|<W_0/2, \\ 0, & |\omega|>W_0/2. \end{array} \right.
\end{equation*}

\end{lemma}

\begin{proof}
See Appendix~\ref{Proof:PSD}.
\end{proof}

In \eqref{eq:PSD}, the first two terms are the PSD related to the antenna mainlobe, while the third term is the PSD related to all the sidelobes. According to Lemma~\ref{Lemma:PSD}, there exist non-direct components in the equivalent uplink channel, meaning that the channel still exhibits some time-variation. 

\vspace*{-0.1in}
\subsection{Scaling Law of the Doppler Spread}\label{ssec:DS}
To measure the time variation of the equivalent uplink channel, we further utilize the general metric Doppler spread~\cite{Souden09Robust,Bellili17A}, which is defined as
\begin{equation}\label{eq:DS_def}
  \sigma_{DS} = \left( \frac{\int_{-2\omega_d}^{2\omega_d}\omega^2P(\omega)d\omega}{\int_{-2\omega_d}^{2\omega_d}P(\omega)d\omega} \right)^{\frac{1}{2}}.
\end{equation}

The Doppler spread of \eqref{eq:DS_def} stands for the standard deviation of the DFO. The deduction of the Doppler spread requires not only the range of the DFOs but also the shape of the channel's PSD. The characterization of the time-varying channel is directly related to this information~\cite{Souden09Robust,Bellili17A}.
Note that the analyzed Doppler spread in \cite{You17BDMA} is defined as the difference between the maximum and minimum DFOs. In contrast, the definition in \eqref{eq:DS_def} takes into account the channel PSD that reflects different contribution of each frequency component on channel variation. Therefore, \eqref{eq:DS_def} would provide a more comprehensive definition of Doppler spread for time-varying channels.

\begin{lemma}\label{Lemma:DS}
The Doppler spread of the equivalent uplink channel is expressed as
\begin{equation}\label{eq:DS}
  \sigma_{DS} = \left( \Gamma/\Lambda \right)^{\frac{1}{2}},
\end{equation}
where
\begin{equation}\label{eq:Lambda}
  \Lambda = C_0 + 4{C_1}{W_0} + \sum_{\substack{I_\mathrm{min}\leq i\leq I_\mathrm{max} \\ i\neq\{-1,0\}}} 2{D_i}{\bar{C}_i}{W_0},
\end{equation}
\begin{align}\label{eq:Gamma}
  \Gamma =& \frac{4{C_1}{W_0^3}}{3}-\frac{8{C_1}{W_0}}{t_0^2} \nonumber\\ +& \!\!\sum_{\substack{I_\mathrm{min}\leq i\leq I_\mathrm{max} \\ i\neq\{-1,0\}}}\!\! \left( \frac{4{D_i}{\bar{C}_i}{W_0^3}}{6}-\frac{{D_i}{\bar{C}_i}{W_0}}{{t}_0^2}+2{D_i}{\bar{C}_i}{W_0}{W_i^2}\right).
\end{align}

\end{lemma}

\begin{proof}
See Appendix~\ref{Proof:DS}.
\end{proof}

In high-mobility communications, the maximum DFO $f_d$ that relates to the moving speed will directly affect the system performance. Moreover, when applying a large-scale antenna array, the number of the transmit antennas $M$ is also a critical parameter. To further show the explicit relationship between the Doppler spread and $f_d$ or $M$, we have the following lemma.

\begin{lemma}\label{Lemma:DS_App}
The Doppler spread $\sigma_{DS}$ is proportional to the maximum DFO $f_d$, and when $M$ is sufficiently large, $\sigma_{DS}$ can be approximated as
\begin{equation}\label{eq:DS_App}
\sigma_{DS} \simeq 2\pi\kappa f_d (\ln(4M))^{-\frac{1}{2}} M^{-\frac{1}{2}},
\end{equation}
where $\kappa$ is a coefficient independent of both $f_d$ and $M$. Thus, the Doppler spread $\sigma_{DS}$ decreases approximately as $1/\sqrt{M}$ with increasing $M$ when $M$ is sufficiently large.
\end{lemma}

%

\begin{proof}
Bearing in mind that $\theta_t$ is determined by \eqref{eq:theta_t}, since $\theta_t$ is small, we have ${\Delta_m}{\theta_t}\simeq {\Delta_m}\sin{\theta_t}=\frac{\lambda}{Md}$. Then, $C_0$ is simplified as
\begin{equation}
  C_0 = \frac{8{\Delta_m}{\theta_t}}{\pi}\simeq\frac{8\lambda}{\pi Md}.
\end{equation}

When $M$ is large, $\Delta_m$ always satisfies the following equation
\begin{equation}
  \Delta_m = \arccos\left(1-\frac{\lambda}{Md}\right) \simeq \frac{\lambda}{\sqrt{M}d}.
\end{equation}

Then, we can simply $C_1$ as follows
\begin{align}
  C_1 \simeq & -\frac{2}{w_d}\ln\tan\frac{\sin{\theta_t}}{2} \nonumber\\ =& -\frac{2}{w_d}\ln\tan \frac{\lambda}{2Md\arccos \left( 1-\frac{\lambda}{Md} \right)} \nonumber \\
 \simeq & -\frac{2}{w_d}\ln\tan\frac{\lambda}{2Md\frac{\lambda}{\sqrt{M}d}} \nonumber \\
 \simeq & -\frac{2}{w_d}\ln\frac{1}{2\sqrt{M}} = \frac{1}{w_d}\ln(4M).
\end{align}

It is easy to find that $D_i=D_{-(i+1)}$ and $\bar{C}_i=\bar{C}_{-(i+1)}$ for $i>1$ (See Appendix~\ref{Proof:Eq2}), and $I_\mathrm{max}=-I_\mathrm{min}-1$ holds. Therefore, we only consider $i>0$ in the following analysis. Since $\bar{\theta}_t$ approaches zero if $M$ tends to infinite, we have
\begin{equation}
\bar{C}_i \simeq \frac{1}{w_d}\bar{F}(M,i),
\end{equation}
where the function $\bar{F}(M,i)$ is defined as
\begin{equation}
\bar{F}(M,i) = \int_{0}^{b_i}{\frac{1}{\sqrt{1-{\left( \cos\theta-u_i \right)}^2}} d\theta}.
\end{equation}

Rewrite the function $\bar{F}(M,i)$ as
\begin{align}
  \bar{F}(M,i) 
  =& \int_{0}^{b_i}{\frac{\sin\theta}{\sqrt{1-\cos^2\theta}\sqrt{1-(\cos\theta-u_i)^2}} d\theta} \nonumber\\
  =& \int_{u_i-1}^{1}{\frac{1}{\sqrt{1-y^2}\sqrt{1-(y-u_i)^2}} dy}.
\end{align}

According to the Equation (3.147-4) in~\cite{Gradshteyn07Table}, we have
\begin{equation}\label{eq:F}
  \bar{F}(M,i) = F(\mu_i) = \int_{0}^{\pi/2}\frac{d\xi}{\sqrt{1-\mu_i^2\sin^2\xi}},
\end{equation}
where $\mu_i=\sqrt{1-u_i^2/4}$ with $u_i=\frac{(2i+1)\lambda}{2Md}$. Note that the function $F(\mu_i)$ in \eqref{eq:F} is the standard complete elliptic integral of the first kind \cite{Gradshteyn07Table,Carlson10Elliptic}. It can be further expressed as a power series
\begin{equation}\label{eq:F4}
  F(\mu_i) = \frac{\pi}{2}\sum_{n=0}^\infty \Upsilon(n)\mu_i^{2n},
\end{equation}
where $\Upsilon(n) = \left(\frac{(2n)!}{2^{2n}(n!)^2}\right)^2$
is the coefficient for the power series expansion.

We further express $D_i$ into a piecewise linear function, which is comprised of the following three parts
\begin{equation}\label{eq:Di}
D_i=\left\{ \begin{array}{ll}
   \frac{4}{(2i+1)^2\pi^2}, & 1\leq i\leq I_1,  \\
   \frac{1}{M^2}, & I_1 < i \leq I_2,  \\
   \frac{4}{(2M-2i-1)^2\pi^2}, & I_2 < i \leq I_\mathrm{max},
\end{array} \right.
\end{equation}
where $I_1=\lfloor\frac{M}{\pi}-\frac{1}{2}\rfloor$, $I_2=\lfloor M(1-\frac{1}{\pi})-\frac{1}{2}\rfloor$, and $I_\mathrm{max}=\lfloor\frac{2Md}{\lambda}-\frac{1}{2}\rfloor$.

Now we reconsider $\Lambda$ and $\Gamma$ based on the above observations. First, rewrite $\Lambda$ in \eqref{eq:Lambda} into
\begin{align}\label{eq:Lambda_1}
  \Lambda & = C_0+4{C_1}{W_0}+2\times\sum\limits_{1\leq i \leq I_\mathrm{max}}2{D_i}{\bar{C}_i}{W_0} \nonumber\\
 & \simeq\frac{\lambda}{Md}\left( \frac{8}{\pi}+4\ln \left(4M\right)+\Lambda_1 \right),
\end{align}
where
\begin{align}\label{eq:Lambda1_1}
  \Lambda_1 = \sum\limits_{1\leq i\leq I_\mathrm{max}}4{D_i}F(\mu_i).
\end{align}

According to the approximation for $D_i$ in \eqref{eq:Di}, $\Lambda_1$ can be further divided into three terms, that is
\begin{align}\label{eq:Lambda1_2}
  \Lambda_1 = \Lambda_{1,1}+\Lambda_{1,2}+\Lambda_{1,3},
\end{align}
where
\begin{equation}\label{eq:Lambda11}
  \Lambda_{1,1} = \sum\limits_{1\leq i\leq I_1} \frac{16}{(2i+1)^2\pi^2}\cdot \frac{\pi}{2}\sum_{n=0}^{\infty} \Upsilon(n)\mu_i^{2n},
\end{equation}

\begin{equation}\label{eq:Lambda12}
  \Lambda_{1,2} = \sum\limits_{I_1< i\leq I_2} \frac{4}{M^2}\cdot \frac{\pi}{2}\sum_{n=0}^{\infty} \Upsilon(n)\mu_i^{2n},
\end{equation}

\begin{equation}\label{eq:Lambda13}
  \Lambda_{1,3} = \sum\limits_{I_2< i\leq I_\mathrm{max}} \frac{16}{(2M-2i-1)^2\pi^2}\cdot \frac{\pi}{2}\sum_{n=0}^{\infty} \Upsilon(n)\mu_i^{2n}.
\end{equation}

As discussed in Appendix~\ref{Proof:Lambda}, we have the following observations:
\begin{itemize}
\item $\Lambda_{1,1}$ increases no faster than $\frac{2}{\pi}\ln{2M}$ when increasing $M$;

\item Both $\Lambda_{1,2}$ and $\Lambda_{1,3}$ decrease approximately as $1/M$ when increasing $M$, thus they can be neglected when $M$ is sufficiently large.
\end{itemize}


Then, we rewrite $\Gamma$ in \eqref{eq:Gamma} as
\begin{align}\label{eq:Gamma_2}
  \Gamma = & \frac{4{C_1}{W_0^3}}{3}-\frac{8{C_1}{W_0}}{t_0^2} \nonumber\\ 
  + & 2\times\sum_{1\leq i\leq I_\mathrm{max}}\!\! \left( \frac{{D_i}{\bar{C}_i}{W_0^3}}{6}-\frac{{D_i}{\bar{C}_i}{W_0}}{{t}_0^2}+2{D_i}{\bar{C}_i}{W_0}{W_i^2}\right) \nonumber\\
 \simeq & {w_d}^2{\left(\frac{\lambda}{Md}\right)^3}\left( (\frac{4}{3}-\frac{8}{\pi^2})\ln(4M)+\Gamma_1+\Gamma_2 \right),
\end{align}
where
\begin{align}
  \Gamma_1 = \sum\limits_{1\le i\le I_{max}} (\frac{1}{3}-\frac{2}{\pi^2}){D_i}F(\mu_i),
\end{align}
\begin{align}
  \Gamma_2 = \sum\limits_{1\le i\le I_{max}} (2i+1)^2{D_i}F(\mu_i).
\end{align}

Similarly to the analysis for $\Lambda_1 $, we directly arrive at
\begin{align}
  \Gamma_1 < \frac{1}{2\pi}(\frac{1}{3}-\frac{2}{\pi^2})\ln{2M},
\end{align}
which shows that $\Gamma_1$ increases no faster than $\frac{1}{2\pi}(\frac{1}{3}-\frac{2}{\pi^2})\ln{2M}$ when increasing $M$.

We further divide $\Gamma_2$ into three terms according to the values of $D_i$ in \eqref{eq:Di}, that is
\begin{align}\label{eq:Gamma2}
  \Gamma_2 = \Gamma_{2,1}+\Gamma_{2,2}+\Gamma_{2,3},
\end{align}
where
\begin{equation}\label{eq:Gamma21}
  \Gamma_{2,1} = \sum\limits_{1\leq i\leq I_1} \frac{4(2i+1)^2}{(2i+1)^2\pi^2}\cdot \frac{\pi}{2}\sum_{n=0}^{\infty} \Upsilon(n)\mu_i^{2n},
\end{equation}

\begin{equation}\label{eq:Gamma22}
  \Gamma_{2,2} = \sum\limits_{I_1< i\leq I_2} \frac{(2i+1)^2}{M^2}\cdot \frac{\pi}{2}\sum_{n=0}^{\infty} \Upsilon(n)\mu_i^{2n},
\end{equation}

\begin{equation}\label{eq:Gamma23}
  \Gamma_{2,3} = \sum\limits_{I_2< i\leq I_\mathrm{max}} \frac{4(2i+1)^2}{(2M-2i-1)^2\pi^2}\cdot \frac{\pi}{2}\sum_{n=0}^{\infty} \Upsilon(n)\mu_i^{2n}.
\end{equation}

We have the following observations (See Appendix~\ref{Proof:Gamma2}):

\begin{itemize}
\item $\Gamma_{2,1}$ increases no faster than $M$ when increasing $M$;

\item $\Gamma_{2,2}$ is proportional to $M$ and increases faster than $\frac{6}{\pi^2}(1-2/\pi)\cdot M$ when increasing $M$;

\item $\Gamma_{2,3}$ is proportional to $M$ and increases faster than $\frac{6}{\pi^2}(2d/\lambda-1+1/\pi)\cdot M$ when increasing $M$.
\end{itemize}

Finally, by substituting \eqref{eq:Lambda_1} and \eqref{eq:Gamma_2} into \eqref{eq:DS}, we arrive at
\begin{equation}
  \sigma_{DS} \simeq \frac{2\pi\lambda f_d}{Md}\left( \frac{\left(\frac{4}{3}-\frac{8}{\pi^2}\right)\ln(4M)+\Gamma_1+\Gamma_2}{\frac{8}{\pi}+4\ln(4M)+\Lambda_1} \right)^\frac{1}{2}.
\end{equation}

It is obvious that the Doppler spread $\sigma_{DS}$ is proportional to $f_d$.
According to the analysis for $\Lambda_1$, $\Gamma_1$ and $\Gamma_2$, we have the following conclusions when $M$ is sufficiently large:

\begin{itemize}
\item In the denominator, $\Lambda_1$ increases slower than the term $4\ln(4M)$ with increasing $M$, which means that the variation with $M$ is mainly determined by $4\ln(4M)$.

\item In the numerator, $\Gamma_2$ increases as $M$ with increasing $M$. Compared with $\Gamma_2$, the variation of both $(\frac{4}{3}-\frac{8}{\pi^2})\ln(4M)$ and $\Gamma_1$ can be neglected.
\end{itemize}

Based on the above important observations, we further approximate $\sigma_{DS}$ for large value of $M$, that is
\begin{equation}\label{eq:DS_4}
\sigma_{DS} \simeq 2\pi\kappa f_d (\ln(4M))^{-\frac{1}{2}} M^{-\frac{1}{2}},
\end{equation}
where $\kappa$ is a coefficient independent of $f_d$ and $M$. Thus, we finally arrive at \eqref{eq:DS_App}. Since the function $(\ln(4M))^{-\frac{1}{2}}$ changes quite slowly in the range of large $M$, we can deduce from~\eqref{eq:DS_4} that the Doppler spread $\sigma_{DS}$ decreases approximately as $1/\sqrt{M}$ with increasing $M$ when $M$ is sufficiently large. This completes the proof.
\end{proof}




\section{Simulation Results}

\begin{table}[t!]
	\renewcommand{\arraystretch}{1.2}
	\footnotesize
	\caption{Simulation Parameters\vspace*{-0.0in}}%
	\centering 
	\label{table:SimuPara-UL}
	\begin{tabular}{|l|l|}
		\hline number of subcarriers & 128\\
		\hline carrier frequency & 3GHz\\
		\hline wavelength of carrier wave & $\lambda=0.1\mathrm{m}$ \\
		\hline number of blocks in each frame & 5\\
		\hline duration of each block & $T_b=0.1\mathrm{ms}$\\
		\hline normalized maximum DFO & $f_d T_b=0.1$ for $v=360\mathrm{km/h}$\\
		\hline modulation type & 16QAM\\
		\hline \multirow{2}{*}{antenna configuration} & ULA at both Tx and Rx \\
		\cline{2-2}
		&Tx: 128/256/512/1024, Rx: 4\\
		\hline  \multirow{3}{*} {channel parameter}
		& tap number: 6\\
		\cline{2-2}
		& path in each tap: 64\\
		\cline{2-2}
		& maximum channel delay: 16 \\
		\hline uplink receiver type & MRC-LS \\
		\hline
	\end{tabular}
	\vspace*{-0.0in}
\end{table}

In this section, we evaluate the performance of the proposed uplink transmission scheme through numerical simulations. We consider the frame structure in an OFDM system, where the first block is the training block and the remaining are used for data symbols.
We assume the Jakes' channel model between the BS and the moving terminal as in~\cite{Zheng03Simulation}.
The transmit antenna element spacing is $d_t=0.45\lambda$ and the beamforming network is designed with the interval of $2^\circ$. The other simulation parameters are shown in Table~\ref{table:SimuPara-UL}.

First, we compare the PSD of the equivalent uplink channel expressed in~\eqref{eq:PSD} with the Jakes' channel introduced in~\cite{Jakes94Microwave}, where the maximum DFO is chosen as $f_d=1\mathrm{KHz}$ and the number of transmit antennas is $M=128$. It is observed from Fig.~\ref{fig:PSD} that the PSD of the equivalent uplink channel concentrates around the frequency of zero while the PSD of the Jakes' channel is around the maximum DFO $f_d$. The results demonstrate that the uplink channel is mainly affected by the smaller DFOs after pre-processing at the transmitter. Note that the range for PSD is doubled in the proposed scheme which is caused by the residual DFOs after the DFO compensation. It is also seen that the PSD of the equivalent uplink channel contains mainlobe and sidelobes with the widths of $2W_0$ and $W_0$, respectively. This is not unexpected and has been verified in Lemma~\ref{Lemma:PSD}.

\begin{figure}[t!]
    \centering
    \subfigure[]{\label{fig:PSD_Jakes}\includegraphics[scale=0.55]{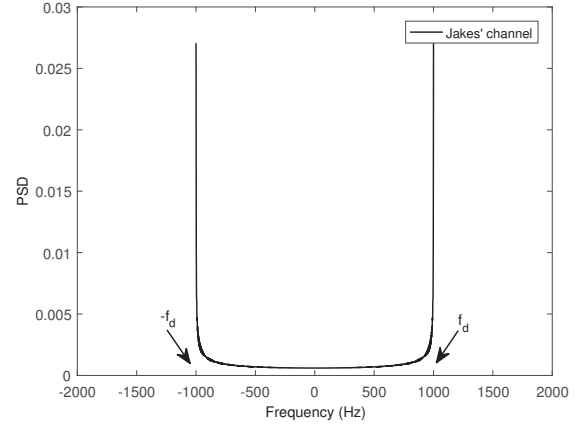}}
    \subfigure[]{\label{fig:PSD_EqULChn}\includegraphics[scale=0.55]{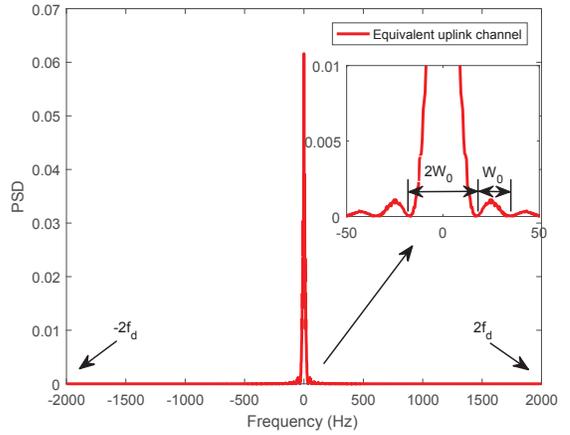}}
    \caption{The PSD comparison: (a) conventional Jakes' channel; (b) the equivalent uplink channel.}
    \label{fig:PSD}
    \vspace*{-0.0in}
\end{figure}

Next, we show the Doppler spread for the equivalent uplink channel as the maximum DFO increases in Fig.~\ref{fig:DS_fd}.
The accurate Doppler spread in dashed curves is calculated by using the definitions in~\eqref{eq:PSD_def} and~\eqref{eq:DS_def} with numerical integration while the analytical approximation in solid curves is obtained by using~\eqref{eq:PSD} and~\eqref{eq:DS}. For comparison, we also include the Doppler spread of the Jakes' channel.
It is clear that the Doppler spread is significantly reduced after the proposed transmit processing. We also make the following observations: 1) The approximation of Doppler spread gets closer to the accurate ones when $M$ gets larger; 2) The Doppler spread of the equivalent channel is proportional to $f_d$ which verifies the analysis in Lemma~\ref{Lemma:DS_App}, and the slope is determined by $M$. With the increase of $M$, the slope of Doppler spread reduces which means that the time variation of channel is mitigated more evidently.

In Fig.~\ref{fig:DS_ant}, we further display the Doppler spread for the equivalent uplink channel with different numbers of transmit antennas $M$. The different markers correspond to different maximum DFOs. The results further verify the accuracy for the approximation of the Doppler spread. Moreover, we can observe that with fixed maximum DFO, the Doppler spread reduces linearly with the increase of $M$ under logarithmic coordinates. The slope is approximately characterized by $-1/2$ as discussed in Lemma~\ref{Lemma:DS_App}.

Finally, we evaluate the symbol error rate (SER) performance of the proposed scheme. The SER performance under different values of $M$ is plotted in Fig.~\ref{fig:SER_SNR}. We assume $f_d=1\mathrm{KHz}$ in this example. For comparison, the conventional transmission schemes with and without DFOs are also included as the benchmark, labeled as Conventional-DFOs and Conventional-NoDFOs, respectively. In conventional transmission, there is no Doppler compensation at both transmitter and receiver while the signal is transmitted through multiple beamforming as in the proposed scheme. Moreover, the conventional time-invariant channel estimation based on least square (LS) is performed at the receiver in all schemes.
The results demonstrate the effectiveness of the proposed scheme. Especially, our scheme outperforms Conventional-DFOs dramatically since Conventional-DFOs suffers from high Doppler spread. From this figure, the SER performance of the proposed scheme gets closer to Conventional-NoDFOs when more antennas are configured on the RS. We can conclude that when the number of transmit antennas is sufficiently large, the BS can completely neglect the time variation of channel and exploit the conventional channel estimation and equalization methods.



\begin{figure}[t]
\centering
\includegraphics[scale=0.55]{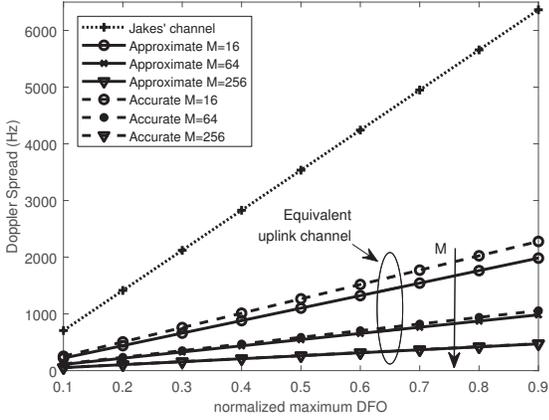}
\caption{The Doppler spread of the equivalent uplink channel as the maximum DFO increases. The X-axis is the normalized maximum DFO by the length of the OFDM block, that is, $f_dT_b$.}
\label{fig:DS_fd}
\vspace*{-0.0in}
\end{figure}

\begin{figure}[t]
\centering
\includegraphics[scale=0.55]{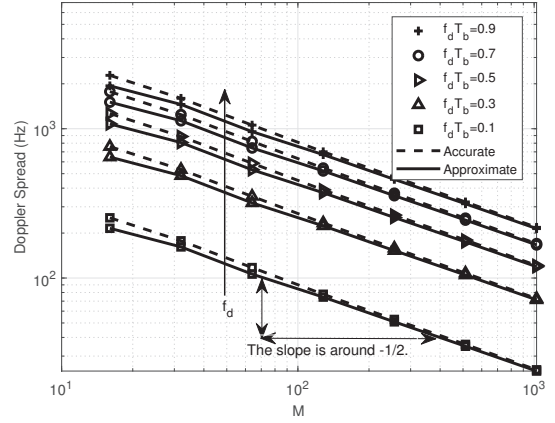}
\caption{The Doppler spread of the equivalent uplink channel as the number of transmit antennas increases.}
\label{fig:DS_ant}
\vspace*{-0.0in}
\end{figure}

\begin{figure}[t]
\centering
\includegraphics[scale=0.55]{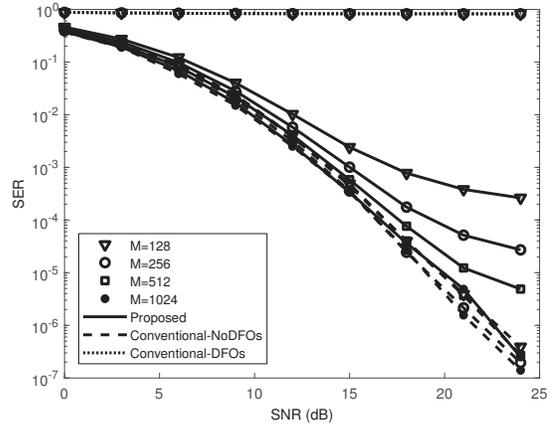}
\caption{The SER performance comparison as the number of transmit antennas increases. The solid curves correspond to the SER performance of the proposed scheme.}
\label{fig:SER_SNR}
\vspace*{-0.0in}
\end{figure}


\section{Conclusion}
In this paper, we have proposed a new transmitter design scheme for high-mobility OFDM uplink transmission, where there exist a large number of DFOs. A massive ULA was configured at the transmitter to provide multiple beamforming branches with high-spatial resolution. The DFO compensation and transmit beamforming were adopted to mitigate the effect of the DFOs previously. We found that the equivalent uplink channel can be considered as time-invariant after the transmit pre-processing and the conventional channel estimation methods can be used at the receiver to recover the transmitted data. Especially, through the analysis of the Doppler spread for the equivalent uplink channel, we derived the asymptotic scaling law of channel variation. The simulation results have shown the effectiveness of the proposed scheme. 


\appendices
\section{Proof of Lemma~\ref{Lemma:ChnCorr}}{\label{Proof:ChnCorr}}
Let us define the autocorrelation for the equivalent fading channel as
\begin{equation}\label{eq:R_def}
R_{\tilde{g}\tilde{g}}(t,\tau) = \mathrm{E}\{\tilde{g}(t)\tilde{g}^*(t+\tau)\}.
\end{equation}

After substituting~\eqref{eq:g_chn_bf} into~\eqref{eq:R_def}, we have
\begin{align}
R_{\tilde{g}\tilde{g}}(t,\tau) & = E_0^2 \int_0^{\pi} \int_0^{\pi} \int_0^{\pi} \int_0^{\pi} \mathrm{E}\left\{\alpha(\tilde{\theta})\alpha^*(\tilde{\theta}')\right\} \nonumber\\
& \cdot \left\{G(\theta,\tilde{\theta})G^*(\theta',\tilde{\theta}') e^{j\omega_d{t}y(\theta,\tilde{\theta})} e^{-j\omega_d(t+\tau)y(\theta',\tilde{\theta}')}\right\} \nonumber\\
& \cdot \mathrm{E}\left\{e^{j\varphi(\theta,\tilde{\theta})-j\varphi(\theta',\tilde{\theta}')}\right\} d\tilde{\theta} d\theta d\tilde{\theta}' d\theta'.
\end{align}

Since the random phase $\phi(\tilde{\theta})$ and $\phi'(\theta)$ are statistically independent and are randomly selected between $0$ and $2\pi$, one can obtain: $\mathrm{E}\{e^{j\varphi(\theta,\tilde{\theta})-j\varphi(\theta',\tilde{\theta}')}\}= \mathrm{E}\{e^{j(\phi(\tilde{\theta})-\phi(\tilde{\theta}'))}\} \mathrm{E}\{e^{j(\phi'(\theta)-\phi'(\theta'))}\}$, $\mathrm{E}\{e^{j(\phi(\tilde{\theta})-\phi(\tilde{\theta}'))}\} = \left\{\! \begin{array}{ll} 0 & \mathrm{if}\ \tilde{\theta}\neq\tilde{\theta}' \\ 1 & \mathrm{if}\ \tilde{\theta}=\tilde{\theta}' \end{array} \right.$, and $\mathrm{E}\left\{e^{j(\phi'(\theta)-\phi'(\theta'))}\right\} = \left\{\! \begin{array}{ll} 0 & \mathrm{if}\ {\theta}\neq{\theta}' \\ 1 & \mathrm{if}\ {\theta}={\theta}' \end{array} \right.$.

According to these facts, we readily arrive at~\eqref{eq:R}. It is observed that the autocorrelation in~\eqref{eq:R} is independent of the time index $t$. This completes the proof.

\section{Proof of Lemma~\ref{Lemma:R_m}}{\label{Proof:R_m}}
It is clear that the maximum beam width of the mainlobe appears when beamforming towards $\theta=0$ or $\pi$, that is, both $\Delta_{l}$ and $\Delta_{u}$ in \eqref{eq:R_m} are upper bounded by
\begin{equation}\label{eq:Delta_m}
  \Delta_{m} = \arccos(1-\frac{\lambda}{Md}).
\end{equation}

Using the first order Taylor expansion, we can approximate $\Delta_{u}$ as
\begin{equation}
  \Delta_{u} \simeq (\theta+\frac{\lambda}{Md}\frac{1}{\sqrt{1-\cos^2\theta}})-\theta=\frac{\lambda}{Md\sin\theta}\triangleq \Delta\theta,
\end{equation}
which is the function of the beamforming angle $\theta$. Similarly, $\Delta_{l}$ is also approximated as $\Delta\theta$. Note that $\Delta\theta$ turns to infinite if $\theta=0$ or $\pi$, which is clearly unreasonable.
To restrict $\Delta\theta$ smaller than the upper bound $\Delta_{m}$, we further adjust $\Delta\theta$ as
\begin{equation}\label{eq:Delta_theta}
  \Delta\theta =\! \left\{\!\! \begin{array}{ll} \frac{\lambda}{Md\sin\theta}, & \theta_t<\theta<\pi-\theta_t,\\ \Delta_{m}, & 0\leq\theta\leq\theta_t \ \mathrm{and}\ \pi-\theta_t\leq\theta\leq\pi, \end{array} \right.
\end{equation}
where $\theta_t$ is the threshold determined by
\begin{equation}\label{eq:theta_t}
  \theta_t: \frac{\lambda}{Md\sin\theta_t}=\arccos(1-\frac{\lambda}{Md}).
\end{equation}


Furthermore, when the number of antennas $M$ is large, the antenna gain of the mainlobe can be expressed as a cosine function approximately, that is
\begin{equation}\label{eq:gain_1}
  |G(\theta,\tilde{\theta})|^2=\frac{\sin^2\pi M\psi} {M^2\sin^2\pi\psi} \simeq \cos^2(\frac{\pi M\psi}{2}),
\end{equation}
where $\psi=\frac{d}{\lambda}y(\theta,\tilde{\theta})$, and $\tilde{\theta}$ lies in the range of the mainlobe.

Based on \eqref{eq:Delta_theta} and \eqref{eq:gain_1}, $R_{\tilde{g}\tilde{g}}^{(m)}(\tau)$ in \eqref{eq:R_m} can be approximated as
\begin{equation}\label{eq:R_m_1}
  R_{\tilde{g}\tilde{g}}^{(m)}(\tau) \!\simeq\! \frac{2}{\pi} \int_0^{\pi}\! \int_{\theta-\Delta\theta}^{\theta+\Delta\theta} \!\!\cos^2(Ay(\theta,\tilde{\theta})) e^{-jBy(\theta,\tilde{\theta})} d\tilde{\theta} d\theta,
\end{equation}
where $A=\frac{\pi Md}{2\lambda}$ and $B=\omega_d \tau$. Denote $\tilde{\theta}=\theta+\varepsilon$, where $\varepsilon$ lies in a small range of $\varepsilon\in[-\Delta\theta,+\Delta\theta]$ if $M$ is large. Using the first order Taylor expansion, $y(\theta,\tilde{\theta})$ can be approximated as: $y(\theta,\tilde{\theta})=\cos(\theta+\varepsilon)-\cos{\theta}\simeq -\varepsilon\sin\theta$. Then, the inner integral in \eqref{eq:R_m_1} turns to
\begin{align}\label{eq:gamma}
  \gamma(\theta) 
  = & \int_{-\Delta\theta}^{\Delta\theta} \cos^2(A\varepsilon\sin\theta) e^{jB\varepsilon\sin\theta} d\varepsilon \nonumber\\
  = & \int_{-\Delta\theta}^{\Delta\theta} \left(\frac{1}{2}\cos(2A\varepsilon\sin\theta)e^{jB\varepsilon\sin\theta}+\frac{1}{2}e^{jB\varepsilon\sin\theta}\right) d\varepsilon \nonumber\\
  = & \frac{-B\sin(B{\Delta\theta}{\sin\theta})\cos(2A{\Delta\theta}{\sin\theta})}{(4A^2-B^2)\sin\theta} \nonumber\\ + & \frac{2A\cos(B{\Delta\theta}{\sin\theta})\sin(2A{\Delta\theta}{\sin\theta})}{(4A^2-B^2)\sin\theta} + \frac{\sin(B{\Delta\theta}{\sin\theta})}{B\sin\theta}.
\end{align}

According to the value of $\Delta\theta$ in different ranges given in \eqref{eq:Delta_theta}, we obtain
\begin{align}\label{eq:R_m_2}
  R_{\tilde{g}\tilde{g}}^{(m)}(\tau) = R_{1}^{(m)}(\tau)+R_{2}^{(m)}(\tau)+R_{3}^{(m)}(\tau),
\end{align}
where $R_{1}^{(m)}(\tau)=\frac{2}{\pi}\int_0^{\theta_t}\gamma(\theta)d\theta$, $R_{2}^{(m)}(\tau)=\frac{2}{\pi}\int_{\theta_t}^{\pi-\theta_t}\gamma(\theta)d\theta$, and $R_{3}^{(m)}(\tau)=\frac{2}{\pi}\int_{\pi-\theta_t}^{\pi}\gamma(\theta)d\theta$.

Since $\theta_t$ is small with a large number of transmit antennas, we have $\sin\theta\simeq\theta$, then
\begingroup\makeatletter\def\f@size{9.5}\check@mathfonts
\def\maketag@@@#1{\hbox{\m@th\normalsize\normalfont#1}}%
\begin{align}\label{eq:R1_m}
  R_{1}^{(m)}(\tau) & \simeq \frac{2}{\pi}\int_0^{\theta_t} \left( \frac{-B\sin(B{\Delta\theta}\cdot{\theta})\cos(2A{\Delta\theta}\cdot{\theta})}{(4A^2-B^2)\theta}\right.  \nonumber\\
  + & \left. \frac{2A\cos(B{\Delta\theta}\cdot{\theta})\sin(2A{\Delta\theta}\cdot{\theta})}{(4A^2-B^2)\theta} + \frac{\sin(B{\Delta\theta}\cdot{\theta})}{B\theta}\right)  d\theta \nonumber\\
  = & \frac{2}{\pi}\int_0^{\theta_t} \left(  \frac{\sin((2A+B){\Delta\theta}\cdot\theta)}{2(2A+B)\theta}\right. \nonumber\\+& \left. \frac{\sin((2A-B){\Delta\theta}\cdot\theta)}{2(2A-B)\theta} + \frac{\sin(B{\Delta\theta}\cdot\theta)}{B\theta} \right) d\theta \nonumber\\
  \simeq & \frac{(2A+B){\Delta_m}\theta_t}{\pi(2A+B)} + \frac{(2A-B){\Delta_m}\theta_t}{\pi(2A-B)} + \frac{2B{\Delta_m}\theta_t}{\pi B} \nonumber\\
  = & \frac{4\Delta_{m}\theta_t}{\pi}.
\end{align}
\endgroup

Similarly, $R_{3}^{(m)}(\tau)$ is calculated as
\begin{align}\label{eq:R3_m}
R_{3}^{(m)}(\tau) = R_{1}^{(m)}(\tau) \simeq \frac{4\Delta_m\theta_t}{\pi}.
\end{align}

Since ${\Delta\theta}\cdot{\sin\theta}=\frac{\lambda}{Md}$ if $\theta_t<\theta<\pi-\theta_t$, $R_{2}^{(m)}(\tau)$ can be simplified as
\begin{align}\label{eq:R2_m}
  R_{2}^{(m)}(\tau) = & \frac{2}{\pi}\int_{\theta_t}^{\pi-\theta_t} \!\! \left(\frac{B\sin(B\lambda/{Md})}{(4A^2-B^2)\sin\theta} + \frac{\sin(B\lambda/{Md})}{B\sin\theta}\right) d\theta \nonumber\\
  = & \frac{8A^2\sin(B\lambda/{Md})}{\pi B(4A^2-B^2)} (-2\ln\tan\frac{\theta_t}{2}).
\end{align}

After substituting \eqref{eq:R1_m}, \eqref{eq:R3_m} and \eqref{eq:R2_m} into \eqref{eq:R_m_2}, we arrive at \eqref{eq:R_m_3}. This completes the proof.

\section{Proof of Lemma~\ref{Lemma:R_s}}{\label{Proof:R_s}}
Similar to the discussion for the mainlobe, both $\Delta_{l,i}$ and $\Delta_{u,i}$ in \eqref{eq:R_s} are approximated as
\begin{equation}\label{eq:dtheta_i}
\Delta\theta_i = \left\{ \begin{array}{ll} \frac{\lambda}{2Md\sin\theta_i}, & \bar{\theta}_t<\theta_i<\pi-\bar{\theta}_t,\\
\bar{\Delta}_m, & 0\leq\theta_i\leq\bar{\theta}_t \ \mathrm{and}\ \pi-\bar{\theta}_t\leq\theta_i\leq\pi, \end{array} \right.
\end{equation}
where $\bar{\Delta}_m=\arccos(1-\frac{\lambda}{2Md})$, and $\bar{\theta}_t$ is the threshold for the sidelobe which satisfies
\begin{equation}\label{eq:bartheta_t}
\bar{\theta}_t: \frac{\lambda}{2Md\sin\bar{\theta}_t}=\arccos(1-\frac{\lambda}{2Md}).
\end{equation}

Moreover, according to the relationship in \eqref{eq:theta_i}, we can express $y(\theta,\tilde{\theta})$ as
\begin{align}\label{eq:y_s}
y(\theta,\tilde{\theta}) = \cos\tilde{\theta}-\cos\theta_i+u_i = y(\theta_i,\tilde{\theta})+u_i.
\end{align}

We further consider the antenna gain for the $i$th sidelobe. We see from \eqref{eq:Gain} that the maximum amplitude is mainly determined by the denominator while the time variation of the antenna gain is mainly determined by the numerator. For simplicity, we can approximate the antenna gain as
%
\begin{equation}\label{eq:G_s}
|G(\theta,\tilde{\theta})| \simeq \left| \frac{\sin(\frac{\pi Md}{\lambda}y(\theta,\tilde{\theta}))} {M\sin(\frac{\pi d}{\lambda}y(\theta,\theta_i))}\right|=\left|\frac{\cos(\frac{\pi M d}{\lambda}y(\theta_i,\tilde{\theta}))}{M\sin(\frac{(2i+1)\pi}{2M})}\right|,
\end{equation}
where the maximum amplitude of the sidelobe is obtained when $\tilde{\theta}=\theta_i$ in the denominator.

After substituting \eqref{eq:dtheta_i}, \eqref{eq:y_s} and \eqref{eq:G_s} into \eqref{eq:R_s}, we have
\begin{align}\label{eq:R_s_1}
R_{\tilde{g}\tilde{g},i}^{(s)}(\tau) = \frac{2}{\pi}D_ie^{-jW_i\tau} \int_0^{\pi} \gamma_i(\theta_i) d\theta,
\end{align}
where
\begin{align}
\gamma_i(\theta_i) = \int_{\theta_i-\Delta\theta_i}^{\theta_i+\Delta\theta_i} \cos^2(\bar{A}y(\theta_i,\tilde{\theta})) e^{-jBy(\theta_i,\tilde{\theta})} d\tilde{\theta},
\end{align}
with $\bar{A}=2A=\frac{\pi Md}{\lambda}$, $B=\omega_d \tau$, $D_i=\left|\frac{1}{M\sin(\frac{(2i+1)\pi}{2M})}\right|^2$, and $W_i=u_i\omega_d$. Similar to the calculation in \eqref{eq:gamma}, $\gamma_i(\theta_i)$ can be directly obtained as
\begin{align}
\gamma_i(\theta_i) = & \frac{-B\sin(B{\Delta\theta_i}{\sin\theta_i})\cos(2\bar{A}{\Delta\theta_i}{\sin\theta_i})}{(4\bar{A}^2-B^2)\sin\theta_i} \nonumber\\
+ & \frac{2\bar{A}\cos(B{\Delta\theta_i}{\sin\theta_i})\sin(2\bar{A}{\Delta\theta_i}{\sin\theta_i})}{(4\bar{A}^2-B^2)\sin\theta_i} \nonumber\\
+ & \frac{\sin(B{\Delta\theta_i}{\sin\theta_i})}{B\sin\theta_i}.
\end{align}

After replacing the variable $\theta$ by the variable $\theta_i$ as in \eqref{eq:theta_i}, we can rewrite $R_{\tilde{g}\tilde{g},i}^{(s)}(\tau)$ into
\begin{equation}\label{eq:R_s_2}
R_{\tilde{g}\tilde{g},i}^{(s)}(\tau) = \frac{2}{\pi}D_ie^{-jW_i\tau} \int_{a_{i}}^{b_{i}} \frac{\gamma_i(\theta_i) \sin\theta_i}{\sqrt{1-(\cos\theta_i-u_i)^2}}d\theta_i,
\end{equation}
 where the range for $\theta_i$ is determined by $a_i$ and $b_i$. Since $\theta_i$ satisfies the expression in \eqref{eq:theta_i} and both $\theta$ and $\theta_i$ lie between 0 and $\pi$, $a_i$ and $b_i$ can be determined, respectively, as follows
\begin{equation}\label{eq:a_i}
a_i = \left\{ \begin{array}{ll} 0, & i>0, \\ \arccos(1+u_i), & i<-1. \end{array} \right.
\end{equation}
\begin{equation}\label{eq:b_i}
b_i = \left\{ \begin{array}{ll} \arccos(-1+u_i), & i>0, \\ \pi, & i<-1. \end{array} \right.
\end{equation}

It can be seen that both the upper and lower bounds for the integral in \eqref{eq:R_s_2} vary with index $i$. We further express \eqref{eq:R_s_2} in the following two cases.

For $i>0$: We can relax the lower bound $a_i=0$ to $a_i=\bar{\theta}_t$ when the number of transmit antennas is large. In the following, we prove that the upper bound $b_i$ always satisfy $b_i<\pi-\bar{\theta}_t$ when $M$ is large.

\begin{proof}
Since $\arccos\left(1-\frac{\lambda}{2Md}\right) \simeq \frac{\lambda}{d\sqrt{2M}}$ when $M$ is large, we have $\sin(\bar{\theta}_t) = \frac{\lambda}{2Md\arccos(1-\frac{\lambda}{2Md})}\simeq \frac{1}{\sqrt{2M}}$. Then, we obtain $\cos(\pi-\bar{\theta}_t) = -\sqrt{1-\sin^2(\bar{\theta}_t)} \simeq -\sqrt{1-\frac{1}{2M}}$, and $\cos(b_i) = -1+u_i \geq -(1-\frac{3\lambda}{2Md})$. It is easy to prove that $(1-\frac{3\lambda}{2Md})<\sqrt{1-\frac{1}{2M}}$ when $M$ is large, that is, $\cos(b_i)>\cos(\pi-\bar{\theta}_t)$. Then, we arrive at $b_i<\pi-\bar{\theta}_t$. This completes the proof.
\end{proof}

Then, there holds ${\Delta\theta_i}\cdot{\sin\theta_i}=\frac{\lambda}{2Md}$ if $\bar{\theta}_t<\theta_i<b_i$. We obtain
\begingroup\makeatletter\def\f@size{9.5}\check@mathfonts
\def\maketag@@@#1{\hbox{\m@th\normalsize\normalfont#1}}%
\begin{align}\label{eq:R_s_3}
R_{\tilde{g}\tilde{g},i}^{(s)}(\tau) & \simeq \frac{2}{\pi}D_ie^{-jW_i\tau} \int_{\bar{\theta}_t}^{b_{i}} \frac{\gamma_i(\theta_i)\sin\theta_i}{\sqrt{1-(\cos\theta_i-u_i)^2}}d\theta_i \nonumber\\
=& D_ie^{-jW_i\tau} \frac{8\bar{A}^2\sin(B\lambda/{2Md})}{\pi B(4\bar{A}^2-B^2)} \int_{\bar{\theta}_t}^{b_i}\!\! \frac{d\theta_i}{\sqrt{1-(\cos\theta_i-u_i)^2}} \nonumber\\
=& D_i \bar{C}_i e^{-jW_i\tau}\left\lbrace \frac{2\sin(\frac{{W}_0}{2}\tau)}{\pi\tau}\right. \nonumber\\
+ & \left.\frac{\sin(\frac{{W}_0}{2}(\tau+2t_0))}{\pi(\tau+2t_0)} + \frac{\sin(\frac{{W}_0}{2}(\tau-2t_0))}{\pi(\tau-2t_0)} \right\rbrace ,
\end{align}
\endgroup
where $W_0=\frac{\omega_d\lambda}{Md}$, $t_0=\frac{\pi Md}{\omega_d\lambda}$, and $\bar{C}_i =\frac{1}{\omega_d} \int_{\bar{\theta}_t}^{b_{i}}  \frac{d\theta_i}{\sqrt{1-(\cos\theta_i-u_i)^2}}$.

For $i<-1$: Similar to the above analysis, the lower bound $a_i>\bar{\theta}_t$ always holds. We relax the upper bound $b_i=\pi$ to $b_i=\pi-\bar{\theta}_t$ when $M$ is large, then we have
\begingroup\makeatletter\def\f@size{9.5}\check@mathfonts
\def\maketag@@@#1{\hbox{\m@th\normalsize\normalfont#1}}%
\begin{align}\label{eq:R_s_4}
R_{\tilde{g}\tilde{g},i}^{(s)}(\tau) & \simeq \frac{2}{\pi}D_ie^{-jW_i\tau} \int_{a_i}^{\pi-\bar{\theta}_t} \frac{\gamma_i(\theta_i)\sin\theta_i}{\sqrt{1-(\cos\theta_i-u_i)^2}}d\theta_i \nonumber\\
=& D_ie^{-jW_i\tau} \frac{8\bar{A}^2\sin(B\lambda/{2Md})}{\pi B(4\bar{A}^2-B^2)} \int_{a_i}^{\pi-\bar{\theta}_t}\!\!\!\! \frac{d\theta_i}{\sqrt{1-(\cos\theta_i-u_i)^2}} \nonumber\\
=& D_i\bar{C}_i e^{-jW_i\tau}\left\lbrace \frac{2\sin(\frac{{W}_0}{2}\tau)}{\pi\tau}\right. \nonumber\\
+& \left.\frac{\sin(\frac{{W}_0}{2}(\tau+2t_0))}{\pi(\tau+2t_0)} + \frac{\sin(\frac{{W}_0}{2}(\tau-2t_0))}{\pi(\tau-2t_0)} \right\rbrace ,
\end{align}
\endgroup
where $\bar{C}_i = \frac{1}{\omega_d} \int_{a_i}^{\pi-\bar{\theta}_t} \frac{d\theta_i}{\sqrt{1-(\cos\theta_i-u_i)^2}}$.

Finally, combining both \eqref{eq:R_s_3} and \eqref{eq:R_s_4}, we readily arrive at \eqref{eq:R_s_0}. This completes the proof.

%

\section{Proof of Lemma~\ref{Lemma:PSD}}{\label{Proof:PSD}}
According to \eqref{eq:R_1}, the PSD of the equivalent uplink channel is comprised of two terms
\begin{equation}\label{eq:PSD_1}
  P(\omega) = P^{(m)}(\omega)+\sum_{\substack{I_\mathrm{min}\leq i\leq I_\mathrm{max} \\ i\neq\{-1,0\}}}P_i^{(s)}(\omega),
\end{equation}
where $P^{(m)}(\omega) = \int_{-\infty}^{+\infty} R_{\tilde{g}\tilde{g}}^{(m)}(\tau)e^{-j\omega\tau} d\tau$ and $P_i^{(s)}(\omega) = \int_{-\infty}^{+\infty} R_{\tilde{g}\tilde{g},i}^{(s)}(\tau)e^{-j\omega\tau} d\tau$.

According to the theory of the Fourier transformation, we directly arrive at
\begin{equation}\label{eq:PSD_m}
  P^{(m)}(\omega) = C_0\delta(\omega) + 2C_1(1+\cos\omega t_0)X(j\omega),
\end{equation}
\begin{align}\label{eq:PSD_si}
  P_i^{(s)}(\omega) = 2{D_i}\bar{C}_i(1+\cos(2(\omega+W_i)t_0) )\bar{X}(j(\omega+W_i)),
\end{align}
where the definitions of the functions $\delta(\omega)$, $X(j\omega)$ and $\bar{X}(j\omega)$ are illustrated in Lemma \ref{Lemma:PSD}. After substituting both \eqref{eq:PSD_m} and \eqref{eq:PSD_si} into \eqref{eq:PSD_1}, we readily arrive at \eqref{eq:PSD}. This completes the proof.

\section{Proof of Lemma~\ref{Lemma:DS}}{\label{Proof:DS}}
In \eqref{eq:DS_def}, we define $\Lambda = \int_{-2\omega_d}^{2\omega_d}P(\omega)d\omega$ and $\Gamma = \int_{-2\omega_d}^{2\omega_d}\omega^2P(\omega)d\omega$. After substituting \eqref{eq:PSD} into $\Lambda$ and $\Gamma$, and performing some tedious manipulations, we have
\begin{equation}
  {\int_{-2\omega_d}^{2\omega_d}P^{(m)}(\omega)d\omega} = C_0 + 4{C_1}{W_0},
\end{equation}
\begin{equation}
\int_{-2\omega_d}^{2\omega_d} P_i^{(s)}(\omega)d\omega = 2{D_i}{\bar{C}_i}{W_0},
\end{equation}
\begin{equation}
  {\int_{-2\omega_d}^{2\omega_d} \omega^2 P^{(m)}(\omega)d\omega} = \frac{4{C_1}{W_0^3}}{3}-\frac{8{C_1}{W_0}}{t_0^2},
\end{equation}
\begin{equation}
{\int_{-2\omega_d}^{2\omega_d}\!\! \omega^2 P_i^{(s)}(\omega) d\omega} =\! \frac{{D_i}{\bar{C}_i}{W_0^3}}{6}-\frac{{D_i}{\bar{C}_i}{W_0}}{t_0^2}+2{D_i}{\bar{C}_i}{W_0}{W_i^2}.
\end{equation}

Based on the above equations, we directly arrive at \eqref{eq:DS}. This completes the proof.

\section{The proof of $\bar{C}_i=\bar{C}_{-(i+1)}$ for $i>0$}{\label{Proof:Eq2}}
Denote $\theta=\pi-\theta_i$, we can rewrite $\bar{C}_i$ for $i>0$ in \eqref{eq:Ci} as
\begin{align}
  \bar{C}_i =& \frac{1}{\omega_d} \int_{\pi-\arccos(-1+u_i)}^{\pi-\bar{\theta}_t} \frac{d\theta}{\sqrt{1-(-\cos\theta-u_i)^2}} \nonumber\\
  =& \frac{1}{\omega_d} \int_{\arccos(1-u_i)}^{\pi-\bar{\theta}_t} \frac{d\theta}{\sqrt{1-(\cos\theta+u_i)^2}}.
\end{align}

Note that $u_i=-u_{-(i+1)}$ for $i>0$. We can further rewrite $\bar{C}_i$ into
\begin{align}
  \bar{C}_i = \frac{1}{\omega_d} \int_{\arccos(1+u_{-(i+1)})}^{\pi-\bar{\theta}_t} \frac{d\theta}{\sqrt{1-(\cos\theta-u_{-(i+1)})^2}}.
\end{align}

Now, we arrive at $\bar{C}_i=\bar{C}_{-(i+1)}$ for $i>0$. This completes the proof.

\section{Relationship between $\Lambda_{1,1}$, $\Lambda_{1,2}$, $\Lambda_{1,3}$ and $M$}{\label{Proof:Lambda}}
Let us first express the first term $\Lambda_{1,1}$ in \eqref{eq:Lambda1_2} as
\begin{align}
  \Lambda_{1,1}
  = & \sum_{n=0}^{\infty} \sum_{1\leq i\leq I_1} \frac{2\Upsilon(n)}{\pi M^2}\cdot\frac{\left(1-(\frac{i}{M}+\frac{1}{2M})^2\right)^n}{(\frac{i}{M}+\frac{1}{2M})^2} \nonumber\\
  \simeq & \sum_{n=0}^{\infty} \frac{2\Upsilon(n)}{\pi M} \int_{1/M}^{1/\pi} \frac{(1-x^2)^n}{x^2}dx,
\end{align}
where the last equation holds approximately when $M$ is sufficiently large. We also relax $I_1$ as a positive real value. Since it is still hard to observe the relationship between $\Lambda_{1,1}$ and $M$, we denote $\bar{x}=\sqrt{1-x^2}$ and there is
\begin{align}\label{eq:Lambda1_3}
  \Lambda_{1,1} \simeq & \sum_{n=0}^{\infty} \frac{2\Upsilon(n)}{\pi M} \int_{\sqrt{1-1/\pi^2}}^{\sqrt{1-1/M^2}} \frac{\bar{x}^{2n}\bar{x}}{(1-\bar{x}^2)\sqrt{1-\bar{x}^2}}d\bar{x} \nonumber\\
  = & \int_{\sqrt{1-1/\pi^2}}^{\sqrt{1-1/M^2}} \frac{4}{\pi^2 M} \frac{F(\bar{x})\bar{x}}{(1-\bar{x}^2)\sqrt{1-\bar{x}^2}}d\bar{x}.
\end{align}

We relax $\Lambda_{1,1}$ to get the close-form expression for the integral in \eqref{eq:Lambda1_3}, that is
\begin{align}
  \Lambda_{1,1} < & \frac{4}{\pi^2 M} \int_{0}^{\sqrt{1-1/M^2}} \frac{F(\bar{x})\bar{x}}{(1-\bar{x}^2)\sqrt{1-1/M^2-\bar{x}^2}}d\bar{x} \nonumber\\
  = & \frac{4}{\pi^2 M} \cdot \frac{\pi}{4}M \ln\left(\frac{1+\sqrt{1-1/M^2}}{1-\sqrt{1-1/M^2}}\right) \nonumber\\
  = & \frac{1}{\pi}\ln\left(M^2(2+2\sqrt{1-1/M^2}-1/M^2)\right) \nonumber\\
  \simeq & \frac{2}{\pi}\ln{2M},
\end{align}
where the integral in the first equation is the integration of the complete elliptic integral and is directly obtained from the Equation (6.153) in \cite{Gradshteyn07Table}. Now, we see that $\Lambda_{1,1}$ increases no faster than $\frac{2}{\pi}\ln{2M}$ when increasing $M$.

The second term $\Lambda_{1,2}$ in \eqref{eq:Lambda1_2} can also be approximated as
\begin{align}
  \Lambda_{1,2} =& \sum_{n=0}^{\infty}\sum\limits_{I_1< i\leq I_2} \frac{2\pi\Upsilon(n)}{M^2}\left(1-(\frac{i}{M}+\frac{1}{2M})^2\right)^n\nonumber\\
  \simeq & \frac{1}{M}\cdot\sum_{n=0}^{\infty} 2\pi\Upsilon(n) \int_{1/\pi}^{1-1/\pi} (1-x^2)^n dx  \nonumber\\
  = & \frac{4}{M} \int_{1/\pi}^{1-1/\pi} F(\sqrt{1-x^2}) dx.
\end{align}

It is obvious that $\Lambda_{1,2}$ decreases approximately as $1/M$ when increasing $M$. Following the similar discussions, we can also arrive at the same result for $\Lambda_{1,3}$. Therefore, both $\Lambda_{1,2}$ and $\Lambda_{1,3}$ can be neglected under a large value of $M$.

\section{Relationship between $\Gamma_{2,1}$, $\Gamma_{2,2}$, $\Gamma_{2,3}$ and $M$}{\label{Proof:Gamma2}}
We rewrite $\Gamma_{2,1}$ in \eqref{eq:Gamma21} into the following integral expression when $M$ is sufficiently large
\begin{align}
  \Gamma_{2,1} =& \sum_{n=0}^{\infty} \sum_{1\leq i\leq I_1} \frac{2M\Upsilon(n)}{\pi}\frac{1}{M}\left(1-(\frac{i}{M}+\frac{1}{2M})^2\right)^n \nonumber\\
  \simeq & \sum_{n=0}^{\infty} \frac{2M\Upsilon(n)}{\pi} \int_{1/M}^{1/\pi} (1-x^2)^n dx \nonumber\\
  = & \frac{4M}{\pi^2} \int_{1/M}^{1/\pi} F(\sqrt{1-x^2}) dx.
\end{align}

Since $F(\sqrt{1-x^2})>0$ holds when $0<x<1$, $\Gamma_{2,1}$ is upper bounded by
\begin{align}
  \Gamma_{2,1} < \frac{4M}{\pi^2} \int_{0}^{1} F(\sqrt{1-x^2}) dx = M,
\end{align}
where the integral is directly obtained from the Equation (6.141) in \cite{Gradshteyn07Table}. Therefore, $\Gamma_{2,1}$ increases no faster than $M$ when increasing $M$.

Next, rewrite $\Gamma_{2,2}$ in \eqref{eq:Gamma22} into
\begingroup\makeatletter\def\f@size{9.2}\check@mathfonts
\def\maketag@@@#1{\hbox{\m@th\normalsize\normalfont#1}}%
\begin{align}\label{eq:Gamma22_1}
  \Gamma_{2,2} = & \sum_{n=0}^{\infty} \sum_{I_1< i\leq I_2} 2\pi\Upsilon(n) \left(\frac{i}{M}+\frac{1}{2M}\right)^2\left(1-(\frac{i}{M}+\frac{1}{2M})^2\right)^n \nonumber\\
  \simeq & M\cdot\sum_{n=0}^{\infty} 2\pi\Upsilon(n) \int_{1/\pi}^{1-1/\pi} x^2(1-x^2)^n dx.
\end{align}
\endgroup

It is obvious that $\Gamma_{2,2}$ is the linear function of $M$. Since $x\geq 1/\pi$ holds in the range of the integral in \eqref{eq:Gamma22_1}, we further have
\begin{align}
  \Gamma_{2,2} > & M\cdot \sum_{n=0}^{\infty} 2\pi\Upsilon(n) \int_{1/\pi}^{1-1/\pi} (1/\pi)^2(1-x^2)^n dx \nonumber\\
  = & M\cdot \frac{4}{\pi^2} \int_{1/\pi}^{1-1/\pi} F(\sqrt{1-x^2}) dx \nonumber\\
  > & M\cdot \frac{6}{\pi^2}\cdot(1-2/\pi),
\end{align}
where the last inequality holds because the complete elliptic integral satisfies $F(\sqrt{1-x^2})>1.5$ \cite{Gradshteyn07Table}. We see that the slope for the function $\Gamma_{2,2}$ is greater than $\frac{6}{\pi^2}\cdot (1-2/\pi)$.

Similarly to the analysis for $\Gamma_{2,2}$, we rewrite $\Gamma_{2,3}$ into
\begin{align}
  \Gamma_{2,3} =& \sum_{n=0}^{\infty} \sum_{I_2< i\leq I_\mathrm{max}}\!\! \frac{2\Upsilon(n)}{\pi} \frac{(\frac{i}{M}+\frac{1}{2M})^2 \left(1-(\frac{i}{M}+\frac{1}{2M})^2\right)^n}{(1-\frac{i}{M}-\frac{1}{2M})^2} \nonumber\\
  \simeq & M\cdot\sum_{n=0}^{\infty} \frac{2\Upsilon(n)}{\pi} \int_{1-1/\pi}^{2d/\lambda} \frac{x^2(1-x^2)^n}{(1-x)^2} dx \nonumber\\
  > & M\cdot \frac{4}{\pi^2}\int_{1-1/\pi}^{2d/\lambda} F(\sqrt{1-x^2}) dx \nonumber\\
  > & M\cdot \frac{6}{\pi^2}\cdot(2d/\lambda-1+1/\pi),
\end{align}
where the first inequality holds because the range of integration is $1-1/\pi<x<2d/\lambda$, thus $1-x<x$ always satisfies. Therefore, we can conclude that $\Gamma_{2,3}$ is also the linear function of $M$ and the slope is greater than $\frac{6}{\pi^2}\cdot (2d/\lambda-1+1/\pi)$.

\bibliographystyle{IEEEtran}
\bibliography{References}

\end{document}